\documentclass[aps,prl,reprint,longbibliography]{revtex4-1}
\usepackage{amsmath}
\usepackage{amssymb}
\usepackage{amsthm}
\usepackage{tikz}
\usetikzlibrary{decorations.markings,decorations.pathreplacing,calc}
\usetikzlibrary{arrows.meta}
\usepackage{algorithm}
\usepackage{algpseudocode}
\usepackage{thm-restate}
\usepackage{verbatim}
\usepackage[framemethod=tikz]{mdframed}
\usepackage{complexity}
\usepackage{enumitem}
\usepackage{subfig}
\usepackage[utf8]{inputenc}
\usepackage{soul}
\usepackage[colorlinks = true]{hyperref}

\usepackage{xcolor}
\definecolor{darkred}  {rgb}{0.5,0,0}
\definecolor{darkblue} {rgb}{0,0,0.5}
\definecolor{darkgreen}{rgb}{0,0.5,0}

\hypersetup{
  urlcolor   = blue,         
  linkcolor  = darkblue,     
  citecolor  = darkgreen,    
  filecolor  = darkred       
}

\newtheorem{prop}{Proposition}

\newtheorem{claim}{Claim}
\newtheorem{lemma}{Lemma}

\newtheorem{fact}{Fact}

\newtheorem{theorem}{Theorem}
\newtheorem*{theorem*}{Theorem}

\newcommand{\be}{\begin{equation}}
\newcommand{\ee}{\end{equation}}
\newcommand{\ba}{\begin{array}}
\newcommand{\ea}{\end{array}}
\newcommand{\bea}{\begin{eqnarray}}
\newcommand{\eea}{\end{eqnarray}}

\def\a{\alpha}

\def\l{\lambda}
\def\m{\mu}

\def\r{\rho}

\newcommand{\bbR}{\mathbb{R}}

\newcommand{\nn}{\nonumber}

\def\id{{\operatorname{id}}}
\newcommand{\br}[1]{\left(#1\right)}

\DeclareMathOperator{\supp}{supp}

\DeclareMathOperator{\Var}{Var}
\newcommand{\tr}{{\rm tr}}
\newcommand{\Tr}{{\rm tr}}
\newcommand{\bra}[1]{\langle #1|}
\newcommand{\ket}[1]{|#1\rangle}

\newcommand{\ketbra}[2]{|#1\rangle\langle #2|}
\def\nn{\nonumber}

\def\norm#1{ {|\hspace{-.022in}|#1|\hspace{-.022in}|} }

\def\iden{I}

\def\id{I}

\def\Pr{\mathrm{Pr}}
\newcommand{\ot}{\otimes}

\newcommand {\Br} [1] {\ensuremath{ \left[ #1 \right] }}

\newcommand{\expec}{\mathbb{E}}

\newcommand{\enq}{\end{equation}}
\newcommand{\beqst}{\begin{equation*}}
\newcommand{\enqst}{\end{equation*}}
\newcommand{\beqar}{\begin{eqnarray}}
\newcommand{\enqar}{\end{eqnarray}}
\newcommand{\beqarst}{\begin{eqnarray*}}
\newcommand{\enqarst}{\end{eqnarray*}}
\def\beq{\begin{equation}}
\def\eeq{\end{equation}}

\newcommand{\thmref}[1]{Theorem~\ref{thm:#1}}

\begin{document}

\title{Improved approximation algorithms for bounded-degree local Hamiltonians}

\author{Anurag Anshu$^{1}$}
\author{David Gosset$^{2}$}
\author{Karen J. Morenz Korol$^{3}$}
\author{Mehdi Soleimanifar$^{4}$}
\affiliation{$^1$ Department of EECS $\&$ Challenge Institute for Quantum Computation, University of California, Berkeley, USA and Simons Institute for the Theory of Computing, Berkeley, California, USA.}
\affiliation{$^2$ Department of Combinatorics and Optimization and Institute for Quantum Computing, University of Waterloo, Canada}
\affiliation{$^3$ Department of Chemistry, University of Toronto, Canada}
\affiliation{$^4$ Center for Theoretical Physics, Massachusetts Institute of Technology, USA}

\begin{abstract}
We consider the task of approximating the ground state energy of two-local quantum Hamiltonians on bounded-degree graphs. Most existing algorithms optimize the energy over the set of product states. Here we describe a family of shallow quantum circuits that can be used to improve the approximation ratio achieved by a given product state. The algorithm takes as input an $n$-qubit product state $|v\rangle$ with mean energy $e_0=\langle v|H|v\rangle$ and variance $\mathrm{Var}=\langle v|(H-e_0)^2|v\rangle$, and outputs a state with an energy that is lower than $e_0$ by an amount proportional to $\mathrm{Var}^2/n$. In a typical case, we have $\mathrm{Var}=\Omega(n)$ and the energy improvement is proportional to the number of edges in the graph.  When applied to an initial random product state, we recover and generalize the performance guarantees of known algorithms for bounded-occurrence classical constraint satisfaction problems. We extend our results to $k$-local Hamiltonians and entangled initial states.
\end{abstract}

\maketitle

Quantum computers are capable of efficiently computing the dynamics of quantum many-body systems \cite{lloyd1996universal}, and it is anticipated that they can be useful for scientific applications in physics, materials science and quantum chemistry. The extent of the quantum advantage for other important simulation tasks, such as computing low temperature properties of quantum systems, is still unknown. In this paper we consider the task of approximating the ground state energy of local Hamiltonians. Here it is natural to expect some improvement over classical machines which cannot even store the state of such systems efficiently. Indeed, classical methods such as the mean-field or Hartree-Fock approximations do not capture the entanglement structure present in the true ground state.  

Motivated by small quantum computers that may be available in the near future, there has been increased interest in devising algorithms that consume few quantum resources and can be implemented across a wide range of hardware platforms. In this vein, heuristic algorithms for ground state preparation have been proposed based on variationally minimizing the energy over the output states of shallow (low-depth) quantum circuits \cite{peruzzo2014VQE, farhi2014quantum,kandala2017hardware}. Although variational algorithms have been rigorously analyzed for specific problems and some limitations are known \cite{mcclean2018barren,farhi2020quantum,BravyiObstaclesQAOA,bravyi2021classical}, no general treatment of their efficacy exists. Characterizing the advantage offered by shallow quantum circuits and variational quantum algorithms stands as a pressing challenge.

In this paper, we derive rigorous bounds on the performance of shallow quantum circuits in estimating the ground state energy of local Hamiltonians. For simplicity, we state our results for a system of qubits with two-local interactions. In the Supplemental Material, we discuss extensions of our results to $k$-local Hamiltonians. 

To begin, let $G=(V,E)$ be a graph, and consider a Hamiltonian 
\begin{equation}
H=\sum_{\{i,j\}\in E} h_{ij}
\label{eq:ham}
\end{equation}
with $n=|V|$ qubits and nearest-neighbor interactions $h_{ij}$ that act nontrivially only on qubits $\{i,j\}$ at vertices connected by an edge. We assume without loss of generality that $\|h_{ij}\|\leq 1$. We are interested in the problem of approximating the ground energy or smallest eigenvalue $\l_{\min}(H)$ of the Hamiltonian. It will be convenient to instead approximate the largest eigenvalue $\lambda_{\max}(H)$; this convention matches the one used in classical optimization and is without loss of generality, since $\lambda_{\min}(H)=-\lambda_{\max}(-H)$. In the worst case, the problem of estimating the largest eigenvalue $\lambda_{\mathrm{max}}(H)$ of Eq.~\eqref{eq:ham} to within an additive error scaling inverse polynomially with $n$ is believed to be intractable for quantum or classical computers \footnote{In particular, a decision version of this problem is complete for the complexity class QMA which is a quantum generalization of NP}.  Here we consider the approximation task where the goal is to compute an estimate $e\leq \lambda_{\mathrm{max}}(H)$ such that the approximation ratio $r\equiv e/\lambda_{max}(H)$ is as large as possible. We will also be interested in efficient quantum algorithms that prepare states $|\psi\rangle$ with good approximation ratios.

Besides describing local interactions encountered in physics, Hamiltonians of the form  Eq.~\eqref{eq:ham} can encode notable cost-functions considered in computer science and thus provide a physically motivated extension of the classical \emph{approximation algorithm} setting \cite{vazirani2013approximation}. For example, one may consider an Ising Hamiltonian for which $h_{ij}=(I+Z_i Z_j)/2$, where $Z$ is the Pauli operator. This Hamiltonian is classical---that is, diagonal in the computational basis---and computing its maximum eigenvalue is equivalent to finding the Max-Cut of the graph $G$, a well-studied classical optimization problem. More generally, two-local quantum Hamiltonians may involve noncommuting terms such as Heisenberg interactions $h_{ij}=1/4(I-X_iX_j-Y_iY_j-Z_iZ_j)$ (with Pauli $X,Y$ and $Z$ operators); the resulting optimization problem can be viewed as a quantum analogue of Max-Cut \cite{GharibianParekhMaxCut}. Quantum approximation algorithms aim to estimate the largest eigenvalue of such Hamiltonians and have been studied in several previous works. This includes the Heisenberg interactions mentioned above \cite{GharibianParekhMaxCut, anshu2020beyond} and more general settings in which the interaction terms $h_{ij}$ are restricted to be positive semidefinite \cite{gharibian2012approximation, hallgren2020approximation, parekh2020beating}, or traceless \cite{harrow2017extremal, bravyi2019approximation}. 

Despite considerable interest, the ultimate limits of efficient algorithms for quantum approximation algorithms are poorly understood.  Approximation ratios approaching $1$ are only known to be achievable for certain special families of graphs, including lattices or bounded-degree planar graphs using tensor product of $O(1)$-qubit states \cite{bansal2007classical} or high degree graphs using tensor products of single-qubit states \cite{bansal2007classical,gharibian2012approximation,brandao2016product}. In certain cases, one may ascertain limitations on efficient achievable approximation ratios from the classical Probabilistically Checkable Proof (PCP) theorem \cite{AroraPCP,AroraPCP2,DinurPCP}, though stronger and more general limitations may follow from the quantum PCP conjecture if some version of it can be proven \cite{aharonov2013guest}. 

A quantum approximation algorithm typically outputs an estimate of the form $\langle v|H|v\rangle$ where  $|v\rangle$ is a quantum state computed by the algorithm. A central challenge is to understand the structure of quantum states $|v\rangle$ that achieve high approximation ratios in the general case.  Most existing algorithms are based on tensor products of one- or few-qubit states, while Ref.~\cite{anshu2020beyond} also considers states prepared by shallow quantum circuits. In this work we describe conditions under which the performance of such algorithms can be \textit{improved}. We restrict our attention to local Hamiltonians on bounded-degree graphs and consider an improvement strategy based on shallow quantum circuits. 

\paragraph{Improvement of product states} To this end, suppose we are given an $n$-qubit state $|v\rangle$ and a Hamiltonian Eq.~\eqref{eq:ham} defined on a graph $G=(V,E)$ with maximum degree $d\geq 2$. It will be convenient to assume (without loss of generality) that $G$ is $d$-regular---we can ensure this by possibly adding some local terms $h_{ij}$ which are equal to zero. We imagine that $|v\rangle$ may be the output of some approximation algorithm such as the ones described above. Our aim is to efficiently compute a state with energy larger than $\langle v|H|v\rangle$.  Moreover, we would like to increase this energy by an amount proportional to $|E|$ in order to guarantee that the \textit{approximation ratio} is larger by some additive constant. We show that this is possible if the following two conditions hold:

\begin{enumerate}[label=(\roman*)]

\item The variance of the energy, defined by $$\Var_v(H)=\langle v|H^2|v\rangle-\langle v|H|v\rangle^2,$$ satisfies $\Var_v(H)=\Omega(|E|)$ \footnote{i.e. there is a universal constant $c$, such that asymptotically $\Var_v(H)\geq c\cdot|E|$}.

\item The state $|v\rangle$ is a product state. That is, $\ket{v}=|v_1\rangle\otimes |v_2\rangle\otimes \ldots \otimes |v_n\rangle$ where each $|v_i\rangle$ is a single-qubit state.
\end{enumerate}
\begin{theorem}
Given a product state $|v\rangle$, we can efficiently compute a depth-$(d+1)$ quantum circuit $U$ such that the state $|\psi\rangle=U|v\rangle$ satisfies
\begin{equation}
\langle \psi|H|\psi\rangle\geq \langle v|H|v\rangle+\Omega\left(\frac{\mathrm{Var}_v(H)^2}{d^2|E|}\right).
\label{eq:varbound}
\end{equation}
\label{thm:prodvar}
\end{theorem}
\vspace{-2mm}
This result applies broadly to quantum optimization problems, but does not provide any improvement when specialized to the classical setting. To see this, note that condition (i) is not satisfied in the purely classical case where $|v\rangle$ is a computational basis state and $H$ is diagonal in the computational basis. Indeed, we have $\mathrm{Var}_v(H)=0$ whenever $|v\rangle$ is an eigenstate of $H$. On the other hand,  condition (i) is fairly mild in the quantum setting, which can be seen from the following expression for the variance:
\[
\mathrm{Var}_v(H)=\sum_{\{i,j\}\cap \{k,l\}\neq \emptyset} \left(\langle v|h_{ij}h_{kl}|v\rangle-\langle v|h_{ij}|v\rangle\cdot \langle v|h_{kl}|v\rangle\right).
\]
Since $G$ is $d$-regular, the number of terms in the sum is $O(d|E|)$. So condition (i) is satisfied if the sum is proportional to the number of terms appearing in it.

Simple examples demonstrate that neither of the two conditions alone is enough to even guarantee the existence of a state with approximation ratio better than $|v\rangle$ for large regular graphs. Condition (ii) alone is not sufficient because it is possible for a product state to have maximal energy $\lambda_{\mathrm{max}}(H)$ (i.e., this occurs for all classical Hamiltonians). To see that condition (i) is not sufficient, one can consider the Max-Cut Hamiltonian on (say) an even cycle graph, and let $|v\rangle$ be an equal superposition of two eigenstates of $H$, one with maximal energy $|E|$ and one with energy $|E|-\Theta(\sqrt{|E|})$. The resulting state has approximation ratio $1-O(|E|^{-1/2})$ and variance $\mathrm{Var}_v(H)=\Omega(|E|)$. Thus condition (i) is satisfied, but the approximation ratio cannot be improved by an additive constant.

In the special case where $|v\rangle$ achieves the largest energy of any product state, we are able to strengthen the bound Eq.~\eqref{eq:varbound}. We say that the product state $|v\rangle$ is \emph{locally optimal} for $H$ if for any single-qubit Pauli $Q$, we have
\[
\frac{d}{d\phi} \langle v|e^{-i\phi Q} H e^{i\phi Q}|v\rangle \big|_{\phi=0} =0,
\]
or equivalently $\langle v|[Q,H]|v\rangle=0$. As we show in the Supplemental Material, the bound in Eq.~\eqref{eq:varbound} can be improved to $\langle v | H | v \rangle +\Omega(\frac{\mathrm{Var}_v(H)^2}{d |E| })$ for locally optimal states. 

Generally, however, the improvement stated in Eq.~\eqref{eq:varbound} is optimal in the sense that there exists a Hamiltonian $H$ and a product state $|v\rangle$ with $\mathrm{Var}_v(H)=\Theta(|E|)$ for which 
\begin{equation}
\lambda_{\mathrm{max}}(H)-\langle v|H|v\rangle\leq O\left(\frac{\mathrm{Var}_v(H)^2}{d^2 |E|}\right).
\label{eq:optimality}
\end{equation}
For example, Eq.~\eqref{eq:optimality} is satisfied by the Hamiltonian with $h_{ij}=Z_i+Z_j$ on any $d$-regular graph and the product state $|v\rangle=\left(\cos(\theta)|0\rangle+\sin(\theta)|1\rangle\right)^{\otimes n}$, for any $\theta\in (0,\pi/2)$. In this simple case, the right-hand side can be computed exactly and is equal to $\frac{\mathrm{Var}(H)^2}{d^2|E|} \cdot \frac{4\sin^2(\theta)}{\sin^2(2\theta)}$.

To establish Theorem~\ref{thm:prodvar}, we consider a variational family of states obtained from $|v\rangle=\otimes_{i \in V} \ket{v_i}$ by applying a quantum circuit composed of nearest neighbor commuting gates on the interaction graph $G$. In particular, let $P_1,P_2,\ldots, P_n$ be any collection of single-qubit operators such that $\|P_i\|\leq 1$ and 
\[
\bra{v_i}P_i\ket{v_i}=0 \quad \text{for all} \quad i\in V.
\]
Following \cite{anshu2020beyond}, we define the circuit
\begin{align}
V(\vec{\theta}) = \bigotimes_{\{i,j\}\in E} e^{i\theta_{ij}P_iP_j}= e^{i\sum_{\{i,j\}\in E}\theta_{ij}P_iP_j}.\label{eq:rr2}
\end{align}
Here, $\vec{\theta}$ is an array of real parameters $\{\theta_{ij}\}_{\{i,j\}\in E}$. Since by assumption, the interaction graph $G$ is $d$-regular, the quantum circuit $V(\vec{\theta})$ can be implemented with circuit depth $d+1$. It is not hard to see that this variational family includes as a special case the level-$1$ Quantum Approximate Optimization Algorithm (QAOA)  for $2$-local classical Hamiltonians \cite{farhi2014quantum}. 
For a given choice of operators $\{P_i\}_{i\in V}$, the following theorem lower bounds the improvement in the energy after applying the the quantum circuit $V(\vec{\theta})$ to $\ket{v}$.

\begin{theorem}\label{thm:general bound}
Let $\ket{v}$ be a product state and $\ket{\psi}=V(\vec{\theta}) \ket{v}$ be the state prepared by the quantum circuit Eq.~\eqref{eq:rr2}. Define the positive real parameter $\a$ by
\begin{align}
    \alpha=\expec_{\{i,j\}\in E}|\bra{v_i,v_j}[P_iP_j, h_{ij}]\ket{v_i,v_j}|,\label{eq:alpha}
\end{align}
where the expectation is with respect to the uniform distribution over the edges. There is an efficient classical algorithm to select parameters $\vec{\theta}$ satisfying
\begin{align}
\bra{\psi}H\ket{\psi}\geq \bra{v}H\ket{v} +\Omega\big(|E|\alpha^2/d\big).\label{eq:r10}
\end{align}
\label{thm:performance}
\end{theorem}
\vspace{-3mm}

\begin{proof}
 
Write $N_{ij}$ for the set of edges $\{k,\ell\}\in E$ incident to a given edge $\{i,j\}\in E$. The latter edge is included as well, i.e.,  $\{i,j\}\in N_{ij}$. Consider the energy of a term
$$\bra{\psi}h_{ij}\ket{\psi}=\bra{v}V(\vec{\theta})^{\dagger}h_{ij}V(\vec{\theta})\ket{v}.$$ The gates in $V(\vec{\theta})$ which are associated with edges that are not incident with $\{i,j\}$ can be cancelled, leaving $\langle v|V_{ij}^{\dagger} h_{ij} V_{ij}|v\rangle$ where $V_{ij}=\prod_{\{k,\ell\}\in N_{ij}} e^{i\theta_{k\ell}P_kP_\ell}$. Thus
\begin{align}
\label{eq:basisexpand}
\bra{\psi}h_{ij}\ket{\psi}&= \bra{v}h_{ij}\ket{v} +\nonumber\\
&\sum_{m=2}^{\infty} \frac{i^m}{m!}\bra{v}\bigg[\sum_{\{k,\ell\}\in N_{ij}} -\theta_{k\ell}P_kP_{\ell}, h_{ij}\bigg]_m\ket{v}.
\end{align}
Here, $\Br{A,B}_m$ is the $m$-nested commutator $\Br{A,\Br{A,\ldots \Br{A,B}}}$. Using the fact that $\bra{v_k}P_k\ket{v_k}=0$ for all $k$, the $m=1$ term simplifies to
\begin{equation}
\label{eq:1storder}
\sum_{\{k,\ell\}\in N_{ij}} -i\theta_{k\ell}\bra{v}\Br{P_kP_{\ell}, h_{ij}}\ket{v}=-i\theta_{ij}\bra{v}[P_iP_j, h_{ij}]\ket{v}.
\end{equation}
 At this stage, we make the choice
\beq
\label{eq:thetaijchoice}
\theta_{ij}=\theta\cdot \text{sign}\br{-i\bra{v}[P_iP_j, h_{ij}]\ket{v}},
\enq
where the parameter $\theta$ will be determined later. Substituting in  Eq. \eqref{eq:1storder} gives
\beq
\label{eq:1storder1}
\sum_{\{k,\ell\}\in N_{ij}} -i\theta_{k\ell}\bra{v}\Br{P_kP_{\ell}, h_{ij}}\ket{v}=\theta|\bra{v_i,v_j}[P_iP_j, h_{ij}]\ket{v_i,v_j}|.
\enq
For $m>1$, we have
\begin{align}
&\bigg|\bra{v}\bigg[\sum_{\{k,\ell\}\in N_{ij}} -\theta_{k\ell}P_kP_{\ell}, h_{ij}\bigg]_m\ket{v}\bigg|\nonumber\\
&\leq\sum_{\substack{\{k_1,\ell_1\}, \{k_2, \ell_2\},\ldots\\ \{k_m,\ell_m\}\in N_{ij}}}\theta^m\left|\bra{v}\Br{P_{k_1}P_{\ell_1}, \Br{\ldots,\Br{P_{k_m}P_{\ell_m},h_{ij}}}}\ket{v}\right|.\nonumber
\end{align}
The only nonzero terms are those in which the expression $\bra{v_s}P_s\ket{v_s}$ does not appear. To upper bound the number of nonzero terms, we count the number of tuples $(\{k_1,\ell_1\}, \{k_2, \ell_2\},\ldots \{k_m,\ell_m\})$ such that no vertex in $V\setminus \{i,j\}$ appears exactly once. An upper bound is provided in the following.
\begin{claim}
Let $m\geq 2$. The number of ordered tuples of edges $(\{k_1,\ell_1\}, \{k_2, \ell_2),\ldots ,\{k_m,\ell_m\})\in N_{ij}^{\times m}$ in which no vertex in $V\setminus\{i,j\}$ appears exactly once is at most $(2m \sqrt{d})^m$. 
\end{claim}
\begin{proof}
First, we count all such tuples $(\{k_1,\ell_1\}, \{k_2, \ell_2\},\ldots \{k_m,\ell_m\})$ in which the edge $\{i,j\}$ does not appear. Each one can be generated by choosing a tuple of vertices $(v_1, v_2, \ldots v_m)$ incident to $\{i,j\}$ and then specifying a neighbor, either $i$ or $j$, for each of them. An upper bound is obtained by counting the number of tuples $(v_1, v_2, \ldots v_m)$ such that each $v_p$ occurs at least twice and then multiplying by $2^m$. Any tuple $(v_1, v_2, \ldots v_m)$ of this form can be generated as follows. First, for each $i=1,2,\ldots, m$ we choose a color $c(i)\in\{1,2\ldots, m/2\}$. We set $v_k=v_{k'}$ whenever $c(k)=c(k')$. We then assign a neighbor of $i$ or $j$ to each color $\{1,2,\ldots, m/2\}$. Since vertices $i$ and $j$ each have at most $d$ neighbors, we see that the number of tuples $(v_1, v_2, \ldots v_m)$ such that each $v_p$ occurs at least twice is at most $(m/2)^m \cdot (2d)^{m/2}$. The number of tuples of edges $(\{k_1,\ell_1\}, \{k_2, \ell_2),\ldots \{k_m,\ell_m\}\}\in N_{ij}^{\times m}$  in which no vertex in $V\setminus\{i,j\}$ appears exactly once, and the edge $\{i,j\}$ does not occur, is then at most $2^m\cdot (m/2)^m \cdot (2d)^{m/2}$.

In order to account for the appearance of the edge $\{i,j\}$, we fix the number of places $u$ where the edge appears and then count as before for the $m-u$ places. This number is 
$$\sum_{u=0}^m {m\choose u}((m-u)\sqrt{2d})^{m-u} \leq (2m \sqrt{d})^m.$$ 
\end{proof}

Finally, using Eq. \eqref{eq:thetaijchoice} and the fact that $\|h_{ij}\|, \|P_i\|\leq 1$, we can upper bound
$$\theta^m\left|\bra{v}\Br{P_{k_1}P_{\ell_1}, \Br{\ldots,\Br{P_{k_m}P_{\ell_m},h_{ij}}}}\ket{v}\right|\leq (2\theta)^{m}.$$
Thus, the sum of all $m>1$ terms in Eq. \eqref{eq:basisexpand} has magnitude at most
\begin{equation*}
\sum_{m=2}^{\infty}\frac{1}{m!}\br{4m\sqrt{d}}^{m}\theta^m\leq \sum_{m=0}^{\infty}\br{4e\sqrt{d}\theta}^{m+2}\leq 32e^2d\theta^2
\end{equation*}
assuming $\theta\leq \frac{1}{8e\sqrt{d}}$ (where we used the bound $m^m/m!\leq e^m$). Combining with Eqs. (\ref{eq:basisexpand},\ref{eq:1storder1}) and summing over all $\{i,j\}\in E$, we get

\[
\bra{\psi}H\ket{\psi}\geq \bra{v}H\ket{v}+|E|\br{\theta \alpha - 32e^2d\theta^2}.
\]
We may then choose $\theta=O(\alpha/d)$ to get the desired lower bound.
\end{proof}

Let us now see how \thmref{prodvar} is obtained as a consequence of Theorem \ref{thm:performance}. The lower bound~\eqref{eq:r10} applies to any choice of operators $\{P_i\}_{i \in V}$. We will choose these operators in a way that gives the variance bound Eq.~\eqref{eq:varbound}. In the following, for convenience and without loss of generality, we shall work in a local basis in which our initial product state is $|v\rangle\equiv |0^n\rangle$. Our starting point is the observation that the variance of a $2$-local Hamiltonian can be expressed in this basis as 
\[
\mathrm{Var}_v(H)=\langle 0^n|HQ_1H|0^n\rangle+\langle 0^n|HQ_2H|0^n\rangle,
\]
where $Q_t$ is the projector onto computational basis states with Hamming weight $t\in \{1,2\}$. This implies that \begin{align}
    \langle 0^n|HQ_t H|0^n\rangle\geq \mathrm{Var}_v(H)/2 \label{eq:var bund Qt}
\end{align}
for some $t\in \{1,2\}$. Suppose $t=2$ and let $X_i$,$Y_i$, and $Z_i$ be the Pauli operators. We define $\alpha_1$ to be the RHS of Eq.~\eqref{eq:alpha} with $P_i=X_i$ for all $i$, and similarly $\alpha_2$ with $P_i=(X_i+Y_i)/\sqrt{2}$ for all $i$.  By a direct calculation we see that 
\begin{align}
   \a_1&=\frac{2}{|E|}\sum_{\{i,j\}\in E} \left|\mathrm{Im}\left(\langle 11|h_{ij}|00\rangle\right)\right|\nn\\
   \a_2&=\frac{2}{|E|}\sum_{\{i,j\}\in E} \left|\mathrm{Re}\left(\langle 11|h_{ij}|00\rangle\right)\right|
\end{align}
and therefore 
\begin{align*}
\langle 0^n|H Q_2 H|0^n\rangle
=\sum_{\{i,j\}\in E} |\langle 11| h_{ij}|00|^2\leq |E|\left(\frac{\alpha_1+\alpha_2}{2}\right).
\end{align*}
This means $\max\{\a_1,\a_2\}\geq |E|^{-1}\langle 0^n|H Q_2 H|0^n\rangle$ which together with Eq.~\eqref{eq:var bund Qt} implies that when $t=2$, we can efficiently find a series of operators $P_i$ such that the parameter $\alpha$ satisfies $\alpha \geq  (2|E|)^{-1}\mathrm{Var}_v(H)$. By plugging this in Eq.~\eqref{eq:r10}, we obtain $\langle \psi|H|\psi\rangle\geq \langle v|H|v\rangle+\Omega(\frac{\mathrm{Var}_v(H)^2}{d|E|})$. Thus if $t=2$ we obtain a better lower bound than the one claimed in Theorem \ref{thm:prodvar}.  Otherwise, if $t=1$, then a simple calculation (reproduced in the Supplemental Material) shows that one can efficiently compute a product state with energy at least $\langle v|H|v\rangle+\Omega(\frac{\mathrm{Var}(H)^2}{d^2|E|})$. In general, the choice between $t=1$ and $t=2$ can be efficiently determined. Thus we obtain Theorem \ref{thm:prodvar}. In the Supplemental Material, we show that if $|v\rangle$ is locally optimal for $H$, then $\langle 0^n|HQ_1H|0^n\rangle=0$ and $t=2$, so we obtain the better bound described above.

Let us briefly illustrate how these results can be applied to the quantum Max-Cut Hamiltonian considered in  Refs.~\cite{GharibianParekhMaxCut,anshu2020beyond}. The Hamiltonian is built from local terms $h_{ij}=w_{ij} \Pi_{ij}$, where $0\leq w_{ij}\leq 1$ and $\Pi_{ij}=(\iden-X_i X_j - Y_i Y_j-Z_i Z_j)/4$ is the projector onto the antisymmetric state of two qubits. This Hamiltonian has the special feature that any product state $|v\rangle$ is locally optimal, and moreover, we have $|\bra{v_i^{\perp},v_j^{\perp}}h_{ij}\ket{v_i,v_j}|=\bra{v_i,v_j}h_{ij}\ket{v_i,v_j}$. Therefore 
\[
\mathrm{Var}_v(H)=\sum_{\{i,j\}\in E} \langle v|h_{ij}|v\rangle^2\geq |E|^{-1}\langle v|H|v\rangle^2
\]
using Cauchy-Schwarz. We may then efficiently compute a state $\ket{\psi}$ such that 
\begin{align}
   \bra{\psi}H\ket{\psi}\geq  \bra{v}H\ket{v} +\Omega\left(\frac{\bra{v}H\ket{v}^4}{d|E|^3}\right).\label{eq:r1}
\end{align}
We see that if the initial state has approximation ratio $\langle v|H|v\rangle/|E|=r$ then the state $|\psi\rangle$ improves this to $r+\Omega(r^4/d)$ \footnote{A better bound can be obtained by directly computing the parameter $\a$ for a randomized choice of operators $\{P_i\}$. In that case, $\expec \a \geq \Omega(\bra{v}H\ket{v}/|E|)$ which results in an improvement of $\Omega\left(\bra{v}H\ket{v}^2/(d|E|)\right)$}.

The preceding example demonstrates the power of Theorem \ref{thm:prodvar} and shows that for the quantum Max-Cut problem, the approximation ratio of \textit{any} product state can be improved by applying a shallow quantum circuit. For more general two-local Hamiltonians, we can guarantee an improvement in the approximation ratio whenever the condition $\Var_v(H)=\Omega(|E|)$ holds, which we expect for most (but not all) product states.  Below we discuss two natural extensions of our results. First, we ask whether one can improve approximation ratios attained by more general families of quantum states. Along these lines, we provide an extension of Theorem \ref{thm:prodvar} to the more general case where $|v\rangle$ is any state prepared by a quantum circuit of depth $D=O(1)$. Next, we show how one can improve the approximation ratio achieved by a random product state $|v\rangle$. Using Theorem \ref{thm:general bound}, we show that the approximation ratio can be improved by $\Omega(1/d)$ for any Hamiltonian with nontrivial two-local interactions, and by $\Omega(1/\sqrt{d})$ if the interaction graph is triangle-free.

\paragraph{Improvement of bounded-depth states}
Recall that for any $n$-qubit quantum circuit and any qubit $j\in [n]$, we may define the lightcone $\mathcal{L}(j)\subseteq [n]$ which consists of all output qubits that are causally connected to $j$. Define the maximum lightcone size $\ell=\max_{j\in [n]} \mathcal{L}(j)$. We have $\ell\leq 2^{D}$ for any depth $D$ circuit composed of two-qubit gates. 

\begin{theorem}
\label{thm:lowdepthimp}
Let $|v\rangle=W|0^n\rangle$ where $W$ is a quantum circuit with maximum lightcone size $\ell$. There is an efficient classical algorithm that computes a quantum circuit $U$ such that $|\psi\rangle= U|v\rangle$ satisfies
\[
\langle \psi|H|\psi\rangle=\langle v|H|v\rangle+\Omega\left(\frac{\mathrm{Var}(H)^2}{\ell^{10}d^2|E|}\right).
\]
\end{theorem}
For constant depth circuits we have $\ell=O(1)$ and we get the same asymptotic energy improvement as we established previously in Theorem \ref{thm:prodvar} for product states. However, in this case the circuit $U$ that we construct is not constant-depth. In the Supplemental Material, we show that the improvement stated above can also be obtained for states $\ket{v}$ that are the unique ground states of a gapped local Hamiltonian $F$. In that case, $\ell$ is replaced by the locality of the Hamiltonian $F$. Thus, Theorem \ref{thm:lowdepthimp} extends to a broad class of tensor network states (such as PEPS of low bond dimension) that have a gapped parent Hamiltonian. 

The theorem provides limitations on the energy that can be achieved by any state $\ket{v}$ produced by a bounded-depth circuit. In particular, since $\bra{\psi}H\ket{\psi}\leq \lambda_{\max}(H)$, we find that 
\[
\langle v|H|v\rangle\leq \lambda_{\max}(H)-\Omega\left(\frac{\mathrm{Var}(H)^2}{\ell^{10}d^2|E|}\right).
\]
This shows that the approximation ratio achievable by constant-depth states $|v\rangle$ with $\mathrm{Var}(H)=\Omega(|E|)$ is bounded away from $1$. An interesting direction for future work is to explore whether one can use this fact to exhibit new local Hamiltonian systems with the almost-linear NLTS (No Low-energy Trivial States) property \cite{AN21, FreedmanH14}.

\paragraph{Improvement of random assignments}
  Given an instance of a (classical) constraint satisfaction problem, one may consider the trivial algorithm in which each variable is chosen independently and uniformly at random. Remarkably, efficient algorithms which improve over the approximation ratio achieved by this simple strategy are not likely to exist in the general case \cite{HastadInapprox}. On the other hand, for structured cases such as bounded-degree graphs, improvement is possible. In particular, on degree-$d$ graphs, one can efficiently find an assignment satisfying a $\m + \Omega(\frac{1}{d})$ fraction of constraints \cite{HastadBoundedOccurence}. Here $\m$ is the expected fraction of constraints satisfied by a uniformly random assignment. It has been shown that when a degree-$d$ graph is triangle-free, there are efficient ``local'' algorithms that find a binary string satisfying a $\m + \Omega(\frac{1}{\sqrt{d} })$ fraction of constraints by starting with a uniformly random assignment \cite{BarakMORRSTVW15,hastings2019BoundedDepth} or quantum superposition \cite{farhi2015quantum} and then locally updating each bit/qubit as a function of the state of its neighbors.

Below we show that this optimal dependence on $d$ can be recovered and generalized to the local Hamiltonian setting by applying our algorithm in Theorem \ref{thm:performance} to a randomly chosen product state. For randomly chosen $\ket{v}$, the parameter $\alpha$ in Theorem \ref{thm:performance} can be related to the $2$-norm of the quadratic terms in the Pauli expansion of the Hamiltonian. More precisely, for an $n$-qubit operator $O=\sum_{i<j}\sum_{x,y}f^{ij}_{xy}\sigma_x^{i}\otimes\sigma_y^{j}$ where $\sigma_0=I$ and $\{\sigma_1, \sigma_2,\sigma_3\}$ are the Pauli matrices, we define $$\mathrm{quad}(O)=\sum_{i<j}\sum_{x>0,y>0}(f^{ij}_{xy})^2.$$

\begin{theorem}
\label{thm:random}
There is an efficient randomized algorithm which computes a depth-$d+1$ circuit $U$ such that $|\psi\rangle=U|v\rangle$ satisfies 
$$\expec_{v} \bra{\psi} H\ket{\psi} \geq \expec_v \bra{v}H\ket{v} +\Omega\left(\frac{\mathrm{quad}(H)^2}{d|E|}\right).$$
 If the graph is triangle-free then the right-hand side can be replaced with $\expec_v \bra{v}H\ket{v}+\Omega\left(\frac{\mathrm{quad}(H)}{\sqrt{d}}\right)$.
\end{theorem}

The proof of Theorem \ref{thm:random} is provided in the Supplemental Material. We also show that for triangle-free graphs one can efficiently compute \textit{product} states matching the approximation ratios quoted above using a local classical algorithm similar to the ones described in Refs.~ \cite{hastings2019BoundedDepth,BarakMORRSTVW15}. Thus, low depth quantum circuits are not necessary to achieve the asymptotic $\Omega(1/\sqrt{d})$ scaling in this case. Nevertheless, one may take the output product state of such algorithms and improve it further using the shallow quantum circuit from Theorem \ref{thm:prodvar}.

\paragraph{Discussion}
For local Hamiltonian problems on bounded-degree graphs, we showed that the approximation ratio achieved by a product state can be improved by a shallow quantum circuit, assuming a mild condition on its variance. Our quantum algorithm generalizes the level-1 Quantum Approximate Optimization Algorithm (QAOA) and extends its applicability beyond classical cost functions. By applying our algorithm to randomly chosen product states we generalized known algorithms for bounded-occurrence classical constraint satisfaction problems. Our results quantify the improvement that shallow quantum circuits can provide over methods based on product states.

\paragraph{Acknowledgments}
 AA acknowledges support from the NSF QLCI program through grant number OMA-2016245. DG acknowledges the support of the Natural Sciences and Engineering Research Council of Canada through grant number RGPIN-2019-04198, the Canadian Institute for Advanced Research, and IBM Research. KJMK acknowledges support from NSERC Vanier Canada Graduate Scholarship. MS was supported by NSF grant CCF-1729369, a Samsung Advanced Institute of Technology Global Research Cluster and grant number FXQi-RFP-1811A from the Foundational Questions Institute and Fetzer Franklin Fund, a donor advised fund of Silicon Valley Community Foundation. 
\bibliography{References}

\begin{thebibliography}{33}%
\makeatletter
\providecommand \@ifxundefined [1]{%
 \@ifx{#1\undefined}
}%
\providecommand \@ifnum [1]{%
 \ifnum #1\expandafter \@firstoftwo
 \else \expandafter \@secondoftwo
 \fi
}%
\providecommand \@ifx [1]{%
 \ifx #1\expandafter \@firstoftwo
 \else \expandafter \@secondoftwo
 \fi
}%
\providecommand \natexlab [1]{#1}%
\providecommand \enquote  [1]{``#1''}%
\providecommand \bibnamefont  [1]{#1}%
\providecommand \bibfnamefont [1]{#1}%
\providecommand \citenamefont [1]{#1}%
\providecommand \href@noop [0]{\@secondoftwo}%
\providecommand \href [0]{\begingroup \@sanitize@url \@href}%
\providecommand \@href[1]{\@@startlink{#1}\@@href}%
\providecommand \@@href[1]{\endgroup#1\@@endlink}%
\providecommand \@sanitize@url [0]{\catcode `\\12\catcode `\$12\catcode
  `\&12\catcode `\#12\catcode `\^12\catcode `\_12\catcode `\%12\relax}%
\providecommand \@@startlink[1]{}%
\providecommand \@@endlink[0]{}%
\providecommand \url  [0]{\begingroup\@sanitize@url \@url }%
\providecommand \@url [1]{\endgroup\@href {#1}{\urlprefix }}%
\providecommand \urlprefix  [0]{URL }%
\providecommand \Eprint [0]{\href }%
\providecommand \doibase [0]{http://dx.doi.org/}%
\providecommand \selectlanguage [0]{\@gobble}%
\providecommand \bibinfo  [0]{\@secondoftwo}%
\providecommand \bibfield  [0]{\@secondoftwo}%
\providecommand \translation [1]{[#1]}%
\providecommand \BibitemOpen [0]{}%
\providecommand \bibitemStop [0]{}%
\providecommand \bibitemNoStop [0]{.\EOS\space}%
\providecommand \EOS [0]{\spacefactor3000\relax}%
\providecommand \BibitemShut  [1]{\csname bibitem#1\endcsname}%
\let\auto@bib@innerbib\@empty
\bibitem [{\citenamefont {Lloyd}(1996)}]{lloyd1996universal}%
  \BibitemOpen
  \bibfield  {author} {\bibinfo {author} {\bibfnamefont {Seth}\ \bibnamefont
  {Lloyd}},\ }\bibfield  {title} {\enquote {\bibinfo {title} {Universal quantum
  simulators},}\ }\href@noop {} {\bibfield  {journal} {\bibinfo  {journal}
  {Science}\ ,\ \bibinfo {pages} {1073--1078}} (\bibinfo {year}
  {1996})}\BibitemShut {NoStop}%
\bibitem [{\citenamefont {Peruzzo}\ \emph {et~al.}(2014)\citenamefont
  {Peruzzo}, \citenamefont {McClean}, \citenamefont {Shadbolt}, \citenamefont
  {Yung}, \citenamefont {Zhou}, \citenamefont {Love}, \citenamefont
  {Aspuru-Guzik},\ and\ \citenamefont {O’brien}}]{peruzzo2014VQE}%
  \BibitemOpen
  \bibfield  {author} {\bibinfo {author} {\bibfnamefont {Alberto}\ \bibnamefont
  {Peruzzo}}, \bibinfo {author} {\bibfnamefont {Jarrod}\ \bibnamefont
  {McClean}}, \bibinfo {author} {\bibfnamefont {Peter}\ \bibnamefont
  {Shadbolt}}, \bibinfo {author} {\bibfnamefont {Man-Hong}\ \bibnamefont
  {Yung}}, \bibinfo {author} {\bibfnamefont {Xiao-Qi}\ \bibnamefont {Zhou}},
  \bibinfo {author} {\bibfnamefont {Peter~J}\ \bibnamefont {Love}}, \bibinfo
  {author} {\bibfnamefont {Al{\'a}n}\ \bibnamefont {Aspuru-Guzik}}, \ and\
  \bibinfo {author} {\bibfnamefont {Jeremy~L}\ \bibnamefont {O’brien}},\
  }\bibfield  {title} {\enquote {\bibinfo {title} {A variational eigenvalue
  solver on a photonic quantum processor},}\ }\href@noop {} {\bibfield
  {journal} {\bibinfo  {journal} {Nature communications}\ }\textbf {\bibinfo
  {volume} {5}},\ \bibinfo {pages} {4213} (\bibinfo {year} {2014})}\BibitemShut
  {NoStop}%
\bibitem [{\citenamefont {Farhi}\ \emph {et~al.}(2014)\citenamefont {Farhi},
  \citenamefont {Goldstone},\ and\ \citenamefont {Gutmann}}]{farhi2014quantum}%
  \BibitemOpen
  \bibfield  {author} {\bibinfo {author} {\bibfnamefont {Edward}\ \bibnamefont
  {Farhi}}, \bibinfo {author} {\bibfnamefont {Jeffrey}\ \bibnamefont
  {Goldstone}}, \ and\ \bibinfo {author} {\bibfnamefont {Sam}\ \bibnamefont
  {Gutmann}},\ }\bibfield  {title} {\enquote {\bibinfo {title} {A quantum
  approximate optimization algorithm},}\ }\href@noop {} {\bibfield  {journal}
  {\bibinfo  {journal} {arXiv preprint arXiv:1411.4028}\ } (\bibinfo {year}
  {2014})}\BibitemShut {NoStop}%
\bibitem [{\citenamefont {Kandala}\ \emph {et~al.}(2017)\citenamefont
  {Kandala}, \citenamefont {Mezzacapo}, \citenamefont {Temme}, \citenamefont
  {Takita}, \citenamefont {Brink}, \citenamefont {Chow},\ and\ \citenamefont
  {Gambetta}}]{kandala2017hardware}%
  \BibitemOpen
  \bibfield  {author} {\bibinfo {author} {\bibfnamefont {Abhinav}\ \bibnamefont
  {Kandala}}, \bibinfo {author} {\bibfnamefont {Antonio}\ \bibnamefont
  {Mezzacapo}}, \bibinfo {author} {\bibfnamefont {Kristan}\ \bibnamefont
  {Temme}}, \bibinfo {author} {\bibfnamefont {Maika}\ \bibnamefont {Takita}},
  \bibinfo {author} {\bibfnamefont {Markus}\ \bibnamefont {Brink}}, \bibinfo
  {author} {\bibfnamefont {Jerry~M}\ \bibnamefont {Chow}}, \ and\ \bibinfo
  {author} {\bibfnamefont {Jay~M}\ \bibnamefont {Gambetta}},\ }\bibfield
  {title} {\enquote {\bibinfo {title} {Hardware-efficient variational quantum
  eigensolver for small molecules and quantum magnets},}\ }\href@noop {}
  {\bibfield  {journal} {\bibinfo  {journal} {Nature}\ }\textbf {\bibinfo
  {volume} {549}},\ \bibinfo {pages} {242--246} (\bibinfo {year}
  {2017})}\BibitemShut {NoStop}%
\bibitem [{\citenamefont {McClean}\ \emph {et~al.}(2018)\citenamefont
  {McClean}, \citenamefont {Boixo}, \citenamefont {Smelyanskiy}, \citenamefont
  {Babbush},\ and\ \citenamefont {Neven}}]{mcclean2018barren}%
  \BibitemOpen
  \bibfield  {author} {\bibinfo {author} {\bibfnamefont {Jarrod~R}\
  \bibnamefont {McClean}}, \bibinfo {author} {\bibfnamefont {Sergio}\
  \bibnamefont {Boixo}}, \bibinfo {author} {\bibfnamefont {Vadim~N}\
  \bibnamefont {Smelyanskiy}}, \bibinfo {author} {\bibfnamefont {Ryan}\
  \bibnamefont {Babbush}}, \ and\ \bibinfo {author} {\bibfnamefont {Hartmut}\
  \bibnamefont {Neven}},\ }\bibfield  {title} {\enquote {\bibinfo {title}
  {Barren plateaus in quantum neural network training landscapes},}\
  }\href@noop {} {\bibfield  {journal} {\bibinfo  {journal} {Nature
  communications}\ }\textbf {\bibinfo {volume} {9}},\ \bibinfo {pages} {1--6}
  (\bibinfo {year} {2018})}\BibitemShut {NoStop}%
\bibitem [{\citenamefont {Farhi}\ \emph {et~al.}(2020)\citenamefont {Farhi},
  \citenamefont {Gamarnik},\ and\ \citenamefont {Gutmann}}]{farhi2020quantum}%
  \BibitemOpen
  \bibfield  {author} {\bibinfo {author} {\bibfnamefont {Edward}\ \bibnamefont
  {Farhi}}, \bibinfo {author} {\bibfnamefont {David}\ \bibnamefont {Gamarnik}},
  \ and\ \bibinfo {author} {\bibfnamefont {Sam}\ \bibnamefont {Gutmann}},\
  }\href@noop {} {\enquote {\bibinfo {title} {The quantum approximate
  optimization algorithm needs to see the whole graph: A typical case},}\ }
  (\bibinfo {year} {2020}),\ \Eprint {http://arxiv.org/abs/2004.09002}
  {arXiv:2004.09002 [quant-ph]} \BibitemShut {NoStop}%
\bibitem [{\citenamefont {Bravyi}\ \emph {et~al.}(2020)\citenamefont {Bravyi},
  \citenamefont {Kliesch}, \citenamefont {Koenig},\ and\ \citenamefont
  {Tang}}]{BravyiObstaclesQAOA}%
  \BibitemOpen
  \bibfield  {author} {\bibinfo {author} {\bibfnamefont {Sergey}\ \bibnamefont
  {Bravyi}}, \bibinfo {author} {\bibfnamefont {Alexander}\ \bibnamefont
  {Kliesch}}, \bibinfo {author} {\bibfnamefont {Robert}\ \bibnamefont
  {Koenig}}, \ and\ \bibinfo {author} {\bibfnamefont {Eugene}\ \bibnamefont
  {Tang}},\ }\bibfield  {title} {\enquote {\bibinfo {title} {Obstacles to
  variational quantum optimization from symmetry protection},}\ }\href
  {\doibase 10.1103/PhysRevLett.125.260505} {\bibfield  {journal} {\bibinfo
  {journal} {Phys. Rev. Lett.}\ }\textbf {\bibinfo {volume} {125}},\ \bibinfo
  {pages} {260505} (\bibinfo {year} {2020})}\BibitemShut {NoStop}%
\bibitem [{\citenamefont {Bravyi}\ \emph {et~al.}(2021)\citenamefont {Bravyi},
  \citenamefont {Gosset},\ and\ \citenamefont
  {Movassagh}}]{bravyi2021classical}%
  \BibitemOpen
  \bibfield  {author} {\bibinfo {author} {\bibfnamefont {Sergey}\ \bibnamefont
  {Bravyi}}, \bibinfo {author} {\bibfnamefont {David}\ \bibnamefont {Gosset}},
  \ and\ \bibinfo {author} {\bibfnamefont {Ramis}\ \bibnamefont {Movassagh}},\
  }\bibfield  {title} {\enquote {\bibinfo {title} {Classical algorithms for
  quantum mean values},}\ }\href@noop {} {\bibfield  {journal} {\bibinfo
  {journal} {Nature Physics}\ }\textbf {\bibinfo {volume} {17}},\ \bibinfo
  {pages} {337--341} (\bibinfo {year} {2021})}\BibitemShut {NoStop}%
\bibitem [{Note1()}]{Note1}%
  \BibitemOpen
  \bibinfo {note} {In particular, a decision version of this problem is
  complete for the complexity class QMA which is a quantum generalization of
  NP}\BibitemShut {NoStop}%
\bibitem [{\citenamefont {Vazirani}(2013)}]{vazirani2013approximation}%
  \BibitemOpen
  \bibfield  {author} {\bibinfo {author} {\bibfnamefont {Vijay~V}\ \bibnamefont
  {Vazirani}},\ }\href@noop {} {\emph {\bibinfo {title} {Approximation
  algorithms}}}\ (\bibinfo  {publisher} {Springer Science \& Business Media},\
  \bibinfo {year} {2013})\BibitemShut {NoStop}%
\bibitem [{\citenamefont {Gharibian}\ and\ \citenamefont
  {Parekh}(2019)}]{GharibianParekhMaxCut}%
  \BibitemOpen
  \bibfield  {author} {\bibinfo {author} {\bibfnamefont {Sevag}\ \bibnamefont
  {Gharibian}}\ and\ \bibinfo {author} {\bibfnamefont {Ojas}\ \bibnamefont
  {Parekh}},\ }\bibfield  {title} {\enquote {\bibinfo {title} {Almost optimal
  classical approximation algorithms for a quantum generalization of
  {M}ax-{C}ut},}\ }in\ \href@noop {} {\emph {\bibinfo {booktitle}
  {Approximation, randomization, and combinatorial optimization. {A}lgorithms
  and techniques}}},\ \bibinfo {series} {LIPIcs. Leibniz Int. Proc. Inform.},
  Vol.\ \bibinfo {volume} {145}\ (\bibinfo  {publisher} {Schloss Dagstuhl.
  Leibniz-Zent. Inform., Wadern},\ \bibinfo {year} {2019})\ pp.\ \bibinfo
  {pages} {Art. No. 31, 17}\BibitemShut {NoStop}%
\bibitem [{\citenamefont {Anshu}\ \emph {et~al.}(2020)\citenamefont {Anshu},
  \citenamefont {Gosset},\ and\ \citenamefont {Morenz}}]{anshu2020beyond}%
  \BibitemOpen
  \bibfield  {author} {\bibinfo {author} {\bibfnamefont {Anurag}\ \bibnamefont
  {Anshu}}, \bibinfo {author} {\bibfnamefont {David}\ \bibnamefont {Gosset}}, \
  and\ \bibinfo {author} {\bibfnamefont {Karen}\ \bibnamefont {Morenz}},\
  }\bibfield  {title} {\enquote {\bibinfo {title} {Beyond product state
  approximations for a quantum analogue of max cut},}\ }in\ \href@noop {}
  {\emph {\bibinfo {booktitle} {15th Conference on the Theory of Quantum
  Computation, Communication and Cryptography (TQC 2020)}}}\ (\bibinfo
  {organization} {Schloss Dagstuhl-Leibniz-Zentrum f{\"u}r Informatik},\
  \bibinfo {year} {2020})\BibitemShut {NoStop}%
\bibitem [{\citenamefont {Gharibian}\ and\ \citenamefont
  {Kempe}(2012)}]{gharibian2012approximation}%
  \BibitemOpen
  \bibfield  {author} {\bibinfo {author} {\bibfnamefont {Sevag}\ \bibnamefont
  {Gharibian}}\ and\ \bibinfo {author} {\bibfnamefont {Julia}\ \bibnamefont
  {Kempe}},\ }\bibfield  {title} {\enquote {\bibinfo {title} {Approximation
  algorithms for qma-complete problems},}\ }\href@noop {} {\bibfield  {journal}
  {\bibinfo  {journal} {SIAM Journal on Computing}\ }\textbf {\bibinfo {volume}
  {41}},\ \bibinfo {pages} {1028--1050} (\bibinfo {year} {2012})}\BibitemShut
  {NoStop}%
\bibitem [{\citenamefont {Hallgren}\ \emph {et~al.}(2020)\citenamefont
  {Hallgren}, \citenamefont {Lee},\ and\ \citenamefont
  {Parekh}}]{hallgren2020approximation}%
  \BibitemOpen
  \bibfield  {author} {\bibinfo {author} {\bibfnamefont {Sean}\ \bibnamefont
  {Hallgren}}, \bibinfo {author} {\bibfnamefont {Eunou}\ \bibnamefont {Lee}}, \
  and\ \bibinfo {author} {\bibfnamefont {Ojas}\ \bibnamefont {Parekh}},\
  }\bibfield  {title} {\enquote {\bibinfo {title} {An approximation algorithm
  for the max-2-local hamiltonian problem},}\ }in\ \href@noop {} {\emph
  {\bibinfo {booktitle} {Approximation, Randomization, and Combinatorial
  Optimization. Algorithms and Techniques (APPROX/RANDOM 2020)}}}\ (\bibinfo
  {organization} {Schloss Dagstuhl-Leibniz-Zentrum f{\"u}r Informatik},\
  \bibinfo {year} {2020})\BibitemShut {NoStop}%
\bibitem [{\citenamefont {Parekh}\ and\ \citenamefont
  {Thompson}(2020)}]{parekh2020beating}%
  \BibitemOpen
  \bibfield  {author} {\bibinfo {author} {\bibfnamefont {Ojas}\ \bibnamefont
  {Parekh}}\ and\ \bibinfo {author} {\bibfnamefont {Kevin}\ \bibnamefont
  {Thompson}},\ }\href@noop {} {\enquote {\bibinfo {title} {Beating random
  assignment for approximating quantum 2-local hamiltonian problems},}\ }
  (\bibinfo {year} {2020}),\ \Eprint {http://arxiv.org/abs/2012.12347}
  {arXiv:2012.12347 [quant-ph]} \BibitemShut {NoStop}%
\bibitem [{\citenamefont {Harrow}\ and\ \citenamefont
  {Montanaro}(2017)}]{harrow2017extremal}%
  \BibitemOpen
  \bibfield  {author} {\bibinfo {author} {\bibfnamefont {Aram~W}\ \bibnamefont
  {Harrow}}\ and\ \bibinfo {author} {\bibfnamefont {Ashley}\ \bibnamefont
  {Montanaro}},\ }\bibfield  {title} {\enquote {\bibinfo {title} {Extremal
  eigenvalues of local hamiltonians},}\ }\href@noop {} {\bibfield  {journal}
  {\bibinfo  {journal} {Quantum}\ }\textbf {\bibinfo {volume} {1}},\ \bibinfo
  {pages} {6} (\bibinfo {year} {2017})}\BibitemShut {NoStop}%
\bibitem [{\citenamefont {Bravyi}\ \emph {et~al.}(2019)\citenamefont {Bravyi},
  \citenamefont {Gosset}, \citenamefont {K{\"o}nig},\ and\ \citenamefont
  {Temme}}]{bravyi2019approximation}%
  \BibitemOpen
  \bibfield  {author} {\bibinfo {author} {\bibfnamefont {Sergey}\ \bibnamefont
  {Bravyi}}, \bibinfo {author} {\bibfnamefont {David}\ \bibnamefont {Gosset}},
  \bibinfo {author} {\bibfnamefont {Robert}\ \bibnamefont {K{\"o}nig}}, \ and\
  \bibinfo {author} {\bibfnamefont {Kristan}\ \bibnamefont {Temme}},\
  }\bibfield  {title} {\enquote {\bibinfo {title} {Approximation algorithms for
  quantum many-body problems},}\ }\href@noop {} {\bibfield  {journal} {\bibinfo
   {journal} {Journal of Mathematical Physics}\ }\textbf {\bibinfo {volume}
  {60}},\ \bibinfo {pages} {032203} (\bibinfo {year} {2019})}\BibitemShut
  {NoStop}%
\bibitem [{\citenamefont {Bansal}\ \emph {et~al.}(2007)\citenamefont {Bansal},
  \citenamefont {Bravyi},\ and\ \citenamefont {Terhal}}]{bansal2007classical}%
  \BibitemOpen
  \bibfield  {author} {\bibinfo {author} {\bibfnamefont {Nikhil}\ \bibnamefont
  {Bansal}}, \bibinfo {author} {\bibfnamefont {Sergey}\ \bibnamefont {Bravyi}},
  \ and\ \bibinfo {author} {\bibfnamefont {Barbara~M}\ \bibnamefont {Terhal}},\
  }\bibfield  {title} {\enquote {\bibinfo {title} {Classical approximation
  schemes for the ground-state energy of quantum and classical ising spin
  hamiltonians on planar graphs},}\ }\href@noop {} {\bibfield  {journal}
  {\bibinfo  {journal} {arXiv preprint arXiv:0705.1115}\ } (\bibinfo {year}
  {2007})}\BibitemShut {NoStop}%
\bibitem [{\citenamefont {Brandao}\ and\ \citenamefont
  {Harrow}(2016)}]{brandao2016product}%
  \BibitemOpen
  \bibfield  {author} {\bibinfo {author} {\bibfnamefont {Fernando~GSL}\
  \bibnamefont {Brandao}}\ and\ \bibinfo {author} {\bibfnamefont {Aram~W}\
  \bibnamefont {Harrow}},\ }\bibfield  {title} {\enquote {\bibinfo {title}
  {Product-state approximations to quantum states},}\ }\href@noop {} {\bibfield
   {journal} {\bibinfo  {journal} {Communications in Mathematical Physics}\
  }\textbf {\bibinfo {volume} {342}},\ \bibinfo {pages} {47--80} (\bibinfo
  {year} {2016})}\BibitemShut {NoStop}%
\bibitem [{\citenamefont {Arora}\ \emph {et~al.}(1998)\citenamefont {Arora},
  \citenamefont {Lund}, \citenamefont {Motwani}, \citenamefont {Sudan},\ and\
  \citenamefont {Szegedy}}]{AroraPCP}%
  \BibitemOpen
  \bibfield  {author} {\bibinfo {author} {\bibfnamefont {Sanjeev}\ \bibnamefont
  {Arora}}, \bibinfo {author} {\bibfnamefont {Carsten}\ \bibnamefont {Lund}},
  \bibinfo {author} {\bibfnamefont {Rajeev}\ \bibnamefont {Motwani}}, \bibinfo
  {author} {\bibfnamefont {Madhu}\ \bibnamefont {Sudan}}, \ and\ \bibinfo
  {author} {\bibfnamefont {Mario}\ \bibnamefont {Szegedy}},\ }\bibfield
  {title} {\enquote {\bibinfo {title} {Proof verification and the hardness of
  approximation problems},}\ }\href {\doibase 10.1145/278298.278306} {\bibfield
   {journal} {\bibinfo  {journal} {J. ACM}\ }\textbf {\bibinfo {volume} {45}},\
  \bibinfo {pages} {501--555} (\bibinfo {year} {1998})}\BibitemShut {NoStop}%
\bibitem [{\citenamefont {Arora}\ and\ \citenamefont
  {Safra}(1998)}]{AroraPCP2}%
  \BibitemOpen
  \bibfield  {author} {\bibinfo {author} {\bibfnamefont {Sanjeev}\ \bibnamefont
  {Arora}}\ and\ \bibinfo {author} {\bibfnamefont {Shmuel}\ \bibnamefont
  {Safra}},\ }\bibfield  {title} {\enquote {\bibinfo {title} {Probabilistic
  checking of proofs: a new characterization of {NP}},}\ }\href {\doibase
  10.1145/273865.273901} {\bibfield  {journal} {\bibinfo  {journal} {J. ACM}\
  }\textbf {\bibinfo {volume} {45}},\ \bibinfo {pages} {70--122} (\bibinfo
  {year} {1998})}\BibitemShut {NoStop}%
\bibitem [{\citenamefont {Dinur}(2007)}]{DinurPCP}%
  \BibitemOpen
  \bibfield  {author} {\bibinfo {author} {\bibfnamefont {Irit}\ \bibnamefont
  {Dinur}},\ }\bibfield  {title} {\enquote {\bibinfo {title} {The {PCP} theorem
  by gap amplification},}\ }\href {\doibase 10.1145/1236457.1236459} {\bibfield
   {journal} {\bibinfo  {journal} {J. ACM}\ }\textbf {\bibinfo {volume} {54}},\
  \bibinfo {pages} {Art. 12, 44} (\bibinfo {year} {2007})}\BibitemShut
  {NoStop}%
\bibitem [{\citenamefont {Aharonov}\ \emph {et~al.}(2013)\citenamefont
  {Aharonov}, \citenamefont {Arad},\ and\ \citenamefont
  {Vidick}}]{aharonov2013guest}%
  \BibitemOpen
  \bibfield  {author} {\bibinfo {author} {\bibfnamefont {Dorit}\ \bibnamefont
  {Aharonov}}, \bibinfo {author} {\bibfnamefont {Itai}\ \bibnamefont {Arad}}, \
  and\ \bibinfo {author} {\bibfnamefont {Thomas}\ \bibnamefont {Vidick}},\
  }\bibfield  {title} {\enquote {\bibinfo {title} {Guest column: the quantum
  pcp conjecture},}\ }\href@noop {} {\bibfield  {journal} {\bibinfo  {journal}
  {Acm sigact news}\ }\textbf {\bibinfo {volume} {44}},\ \bibinfo {pages}
  {47--79} (\bibinfo {year} {2013})}\BibitemShut {NoStop}%
\bibitem [{Note2()}]{Note2}%
  \BibitemOpen
  \bibinfo {note} {I.e. there is a universal constant $c$, such that
  asymptotically $\protect \Var _v(H)\geq c\cdot |E|$}\BibitemShut {NoStop}%
\bibitem [{Note3()}]{Note3}%
  \BibitemOpen
  \bibinfo {note} {A better bound can be obtained by directly computing the
  parameter $\alpha $ for a randomized choice of operators $\{P_i\}$. In that
  case, $\protect \mathbb {E}\alpha \geq \Omega (\langle v|H|v\rangle /|E|)$
  which results in an improvement of $\Omega \left (\langle v|H|v\rangle
  ^2/(d|E|)\right )$}\BibitemShut {NoStop}%
\bibitem [{\citenamefont {Anshu}\ and\ \citenamefont {Nirkhe}(2021)}]{AN21}%
  \BibitemOpen
  \bibfield  {author} {\bibinfo {author} {\bibfnamefont {Anurag}\ \bibnamefont
  {Anshu}}\ and\ \bibinfo {author} {\bibfnamefont {Chinmay}\ \bibnamefont
  {Nirkhe}},\ }\bibfield  {title} {\enquote {\bibinfo {title} {Circuit lower
  bounds for low-energy states of quantum code hamiltonians},}\ }\href@noop {}
  {\  (\bibinfo {year} {2021})},\ \bibinfo {note} {arXiv
  2011.02044}\BibitemShut {NoStop}%
\bibitem [{\citenamefont {Freedman}\ and\ \citenamefont
  {Hastings}(2014)}]{FreedmanH14}%
  \BibitemOpen
  \bibfield  {author} {\bibinfo {author} {\bibfnamefont {Michael~H.}\
  \bibnamefont {Freedman}}\ and\ \bibinfo {author} {\bibfnamefont {Matthew~B.}\
  \bibnamefont {Hastings}},\ }\bibfield  {title} {\enquote {\bibinfo {title}
  {Quantum systems on non-k-hyperfinite complexes: A generalization of
  classical statistical mechanics on expander graphs},}\ }\href@noop {}
  {\bibfield  {journal} {\bibinfo  {journal} {Quantum Info. Comput.}\ }\textbf
  {\bibinfo {volume} {14}},\ \bibinfo {pages} {144–180} (\bibinfo {year}
  {2014})}\BibitemShut {NoStop}%
\bibitem [{\citenamefont {Haastad}(2001)}]{HastadInapprox}%
  \BibitemOpen
  \bibfield  {author} {\bibinfo {author} {\bibfnamefont {Johan}\ \bibnamefont
  {Haastad}},\ }\bibfield  {title} {\enquote {\bibinfo {title} {Some optimal
  inapproximability results},}\ }\href {\doibase 10.1145/502090.502098}
  {\bibfield  {journal} {\bibinfo  {journal} {J. ACM}\ }\textbf {\bibinfo
  {volume} {48}},\ \bibinfo {pages} {798--859} (\bibinfo {year}
  {2001})}\BibitemShut {NoStop}%
\bibitem [{\citenamefont {Haastad}(2000)}]{HastadBoundedOccurence}%
  \BibitemOpen
  \bibfield  {author} {\bibinfo {author} {\bibfnamefont {Johan}\ \bibnamefont
  {Haastad}},\ }\bibfield  {title} {\enquote {\bibinfo {title} {On bounded
  occurrence constraint satisfaction},}\ }\href {\doibase
  10.1016/S0020-0190(00)00032-6} {\bibfield  {journal} {\bibinfo  {journal}
  {Inform. Process. Lett.}\ }\textbf {\bibinfo {volume} {74}},\ \bibinfo
  {pages} {1--6} (\bibinfo {year} {2000})}\BibitemShut {NoStop}%
\bibitem [{\citenamefont {Barak}\ \emph {et~al.}(2015)\citenamefont {Barak},
  \citenamefont {Moitra}, \citenamefont {O'Donnell}, \citenamefont
  {Raghavendra}, \citenamefont {Regev}, \citenamefont {Steurer}, \citenamefont
  {Trevisan}, \citenamefont {Vijayaraghavan}, \citenamefont {Witmer},\ and\
  \citenamefont {Wright}}]{BarakMORRSTVW15}%
  \BibitemOpen
  \bibfield  {author} {\bibinfo {author} {\bibfnamefont {Boaz}\ \bibnamefont
  {Barak}}, \bibinfo {author} {\bibfnamefont {Ankur}\ \bibnamefont {Moitra}},
  \bibinfo {author} {\bibfnamefont {Ryan}\ \bibnamefont {O'Donnell}}, \bibinfo
  {author} {\bibfnamefont {Prasad}\ \bibnamefont {Raghavendra}}, \bibinfo
  {author} {\bibfnamefont {Oded}\ \bibnamefont {Regev}}, \bibinfo {author}
  {\bibfnamefont {David}\ \bibnamefont {Steurer}}, \bibinfo {author}
  {\bibfnamefont {Luca}\ \bibnamefont {Trevisan}}, \bibinfo {author}
  {\bibfnamefont {Aravindan}\ \bibnamefont {Vijayaraghavan}}, \bibinfo {author}
  {\bibfnamefont {David}\ \bibnamefont {Witmer}}, \ and\ \bibinfo {author}
  {\bibfnamefont {John}\ \bibnamefont {Wright}},\ }\href@noop {} {\enquote
  {\bibinfo {title} {Beating the random assignment on constraint satisfaction
  problems of bounded degree},}\ } (\bibinfo {year} {2015}),\ \Eprint
  {http://arxiv.org/abs/1505.03424} {arXiv:1505.03424 [cs.CC]} \BibitemShut
  {NoStop}%
\bibitem [{\citenamefont {Hastings}(2019)}]{hastings2019BoundedDepth}%
  \BibitemOpen
  \bibfield  {author} {\bibinfo {author} {\bibfnamefont {Matthew~B}\
  \bibnamefont {Hastings}},\ }\bibfield  {title} {\enquote {\bibinfo {title}
  {Classical and quantum bounded depth approximation algorithms},}\ }\href@noop
  {} {\bibfield  {journal} {\bibinfo  {journal} {arXiv preprint
  arXiv:1905.07047}\ } (\bibinfo {year} {2019})}\BibitemShut {NoStop}%
\bibitem [{\citenamefont {Farhi}\ \emph {et~al.}(2015)\citenamefont {Farhi},
  \citenamefont {Goldstone},\ and\ \citenamefont {Gutmann}}]{farhi2015quantum}%
  \BibitemOpen
  \bibfield  {author} {\bibinfo {author} {\bibfnamefont {Edward}\ \bibnamefont
  {Farhi}}, \bibinfo {author} {\bibfnamefont {Jeffrey}\ \bibnamefont
  {Goldstone}}, \ and\ \bibinfo {author} {\bibfnamefont {Sam}\ \bibnamefont
  {Gutmann}},\ }\href@noop {} {\enquote {\bibinfo {title} {A quantum
  approximate optimization algorithm applied to a bounded occurrence constraint
  problem},}\ } (\bibinfo {year} {2015}),\ \Eprint
  {http://arxiv.org/abs/1412.6062} {arXiv:1412.6062 [quant-ph]} \BibitemShut
  {NoStop}%
\bibitem [{\citenamefont {O'Donnell}(2014)}]{o2014booleananalysis}%
  \BibitemOpen
  \bibfield  {author} {\bibinfo {author} {\bibfnamefont {Ryan}\ \bibnamefont
  {O'Donnell}},\ }\href@noop {} {\emph {\bibinfo {title} {Analysis of boolean
  functions}}}\ (\bibinfo  {publisher} {Cambridge University Press},\ \bibinfo
  {year} {2014})\BibitemShut {NoStop}%
\end{thebibliography}%

\widetext
 \clearpage
\begin{center}
\textbf{\large Supplemental Materials}
\end{center}
\setcounter{equation}{0}
\setcounter{figure}{0}
\setcounter{table}{0}
\makeatletter
\renewcommand{\theequation}{S\arabic{equation}}
\renewcommand{\thefigure}{S\arabic{figure}}
\renewcommand{\bibnumfmt}[1]{[S#1]}
\renewcommand{\citenumfont}[1]{S#1}

\section{Improvement of product states}

In this section we provide the full details of the proof of Theorem \ref{thm:prodvar}. It will be convenient to work in a local basis defined by $|v\rangle$, such that $|v\rangle=|0^n\rangle$ and
\[
\mathrm{Var}_v(H)= \langle 0^n|H^2|0^n\rangle-(\langle 0^n|H|0^n\rangle)^2.
\]
For ease of notation we write $\mathrm{Var}(H)=\mathrm{Var}_v(H)$. Recall the quantity $\alpha$ defined in Eq.~\eqref{eq:alpha}:
\begin{align}
    \alpha=\expec_{\{i,j\}\in E}|\bra{v_i,v_j}[P_iP_j, h_{ij}]\ket{v_i,v_j}|,\label{eq:alpha2}
\end{align}
We will use the following proposition.
\begin{prop}
Let $Q_2$ be the projector onto computational basis states with Hamming weight 2. We can efficiently choose operators $\{P_i\}_{i\in V}$ such that 
\begin{align}
\alpha\geq \frac{1}{|E|}\cdot \langle 0^n|H Q_2 H|0^n\rangle.\label{eq:alphaQ2}
\end{align}
\label{prop:q2}
\end{prop}
\begin{proof}
Let $\alpha_1$ be Eq.~\eqref{eq:alpha} with $P_i=X_i$ for all $i$, and let  $\alpha_2$ be Eq.~\eqref{eq:alpha} with $P_i=(X_i+Y_i)/\sqrt{2}$ for all $i$. Direct calculation shows that
\begin{align}
   \a_1&=\frac{2}{|E|}\sum_{\{i,j\}\in E} \left|\mathrm{Im}\left(\langle 11|_{ij}\otimes \langle 0^{n-2}| h_{ij}|0^n\rangle\right)\right|\nn\\
   \a_2&=\frac{2}{|E|}\sum_{\{i,j\}\in E} \left|\mathrm{Re}\left(\langle 11|_{ij}\otimes \langle 0^{n-2}| h_{ij}|0^n\rangle\right)\right|
\end{align}

We can express $\langle 0^n|H Q_2 H|0^n\rangle$ as
\begin{align*}
\langle 0^n|H Q_2 H|0^n\rangle.
&=\sum_{\{i,j\}\in E} |\langle 11|_{ij}\otimes \langle 0^{n-2}| h_{ij}|0^n\rangle|^2\\
&=  \sum_{\{i,j\}\in E} \left(\mathrm{Im}\left(\langle 11|_{ij}\otimes \langle 0^{n-2}| h_{ij}|0^n\rangle\right)\right)^2+\left(\mathrm{Re}\left(\langle 11|_{ij}\otimes \langle 0^{n-2}| h_{ij}|0^n\rangle\right)\right)^2\\
&\leq \sum_{\{i,j\}\in E} \left|\mathrm{Im}\left(\langle 11|_{ij}\otimes \langle 0^{n-2}| h_{ij}|0^n\rangle\right)\right|+\left|\mathrm{Re}\left(\langle 11|_{ij}\otimes \langle 0^{n-2}| h_{ij}|0^n\rangle\right)\right|\\
&=|E|\cdot\frac{\alpha_1+\alpha_2}{2},
\end{align*}
where we used the fact that $\|h_{ij}\|\leq 1$ in going from the second to the third line above.  Now the last line implies that either $\alpha_1$ or $\alpha_2$ achieves the bound from Eq.~\eqref{eq:alphaQ2}. Moreover, the choice of $\a_1$ or $\a_2$ can be efficiently determined. 
\end{proof}

\begin{proof}[Proof of Theorem \ref{thm:prodvar}]
Let $Q_t$ be the projector onto computational basis states with Hamming weight $t\in \{1,2\}$. Since $H$ is two-local we have
\[
\mathrm{Var}(H)=\langle 0^n|HQ_1H|0^n\rangle+\langle 0^n|HQ_2H|0^n\rangle.
\]
Therefore $\langle 0^n|HQ_t H|0^n\rangle\geq \mathrm{Var}(H)/2$ for some $t\in \{1,2\}$. If $t=2$ then we may use Proposition \ref{prop:q2} which gives
\[
\max\{\alpha_1,\alpha_2\}\geq \frac{1}{2|E|}\mathrm{Var}(H).
\]
Combining this with Theorem \ref{thm:performance}, we arrive at
\[
\langle \psi|H|\psi\rangle\geq \langle 0^n|H|0^n\rangle+\Omega\left(\frac{\mathrm{Var}(H)^2}{d|E|}\right) 
\]
which is better than the desired lower bound.

Next suppose $\langle 0^n|HQ_1 H|0^n\rangle\geq \mathrm{Var}(H)/2$. Define
\[
L=\sum_{j=1}^{n} (-1)^{a_j} P_j
\] 
where each $P_j$ is a single-qubit Pauli operator acting nontrivially only on qubit $j$, and $a_j\in \{0,1\}$ is chosen so that
\[
i\langle 0^n|(-1)^{a_j} [P_j,H]|0^n\rangle = |\langle 0^n| [P_j,H]|0^n\rangle|.
\] 
Define $|\theta\rangle=e^{-i\theta L}|0^n\rangle$ where $\theta$ is a real parameter that we will fix later. Then
\[
\langle \theta|H|\theta\rangle=\langle 0^n|H|0^n\rangle+\theta \sum_{j=1}^n|\langle 0^n|[P_j, H]|0^n\rangle| +\mathrm{Err},
\]
where 
\begin{align*}
|\mathrm{Err}| &=\left|\sum_{m\geq 2} \frac{i^m \theta^m}{m!} \langle 0^n|[L,H]_m|0^n\rangle\right|\\
&\leq |E|\sum_{m\geq 2} \frac{\theta^m 4^m}{m!}\\
&\leq 16\theta^2|E| e^{4\theta}.
\end{align*}
In the second line we used the fact that 
\[
[L,h_{ij}]_m=[(-1)^{a_i}P_i+(-1)^{a_j}P_j, h_{ij}]_m
\]
can be expanded as a sum of $2^m$ terms each of norm at most $2^m$. Now define
\begin{equation}
\beta=\frac{1}{|E|} \sum_{j=1}^n|\langle 0^n|[P_j, H]|0^n\rangle|.
\label{eq:beta1}
\end{equation}
and note that since $\|[P_j,H]\|\leq 2d$ for all $j$ we have
\[
\beta\leq \frac{1}{|E|}\sum_{j=1}^{n} 2d\leq 4.
\]
Then
\[
\langle \theta|H|\theta\rangle\geq \langle 0^n|H|0^n\rangle+|E|\left(\theta\beta-16\theta^2e^{4\theta}\right).
\]
Choosing $\theta=\beta/32$ gives
\begin{align}
\langle \theta|H|\theta\rangle & \geq\langle 0^n|H|0^n\rangle+|E|\left(\frac{\beta^2}{32}-\frac{\beta^2}{64} e^{\beta/8}\right)\nonumber\\
& \geq \langle 0^n|H|0^n\rangle+|E|\left(\frac{\beta^2}{32}-\frac{\beta^2}{64} e^{1/2}\right)\nonumber\\
& \geq \langle 0^n|H|0^n\rangle+0.001\cdot|E|\beta^2.\label{eq:betabound1}
\end{align}
Now let $\beta_1$ be given by Eq.~\eqref{eq:beta1} with $P_i=X_i$ for all $i$, and let $\beta_2$ be given by Eq.~\eqref{eq:beta1} with $P_i=Y_i$ for all $i$. Then

\begin{align*}
\frac{\beta_1+\beta_2}{2}& =\frac{1}{2|E|}\sum_{j=1}^{n} \left| \langle 0^n|[X_j, H]|0^n\rangle\right|+\left| \langle 0^n|[Y_j,H]|0^n\rangle\right|\\
&\geq \frac{1}{4d|E|}\sum_{j=1}^{n} | \langle 0^n|[X_j, H]|0^n\rangle|^2+\left| \langle 0^n|[Y_j,H]|0^n\rangle\right|^2\\
&=\frac{1}{4d|E|}\sum_{j=1}^{n} \left| 2\mathrm{Im}\left(\langle\hat{e}_j|H|0^n\rangle\right)\right|^2+\left| 2\mathrm{Re}\left(\langle\hat{e}_j|H|0^n\rangle\right)\right|^2\\
&=\frac{1}{d|E|} \langle 0^n |H Q_1 H|0^n\rangle\\
& \geq \frac{1}{2d|E|}\mathrm{Var}(H).
\end{align*}

Therefore either $\beta_1$ or $\beta_2$ is larger than the RHS above. Plugging this into Eq.~\eqref{eq:betabound1} we arrive at
\[
\langle \theta|H|\theta\rangle \geq \langle 0^n|H|0^n\rangle+\Omega\left(\frac{\mathrm{Var}(H)^2}{d^2|E|}\right).
\]
\end{proof}

Finally, let us discuss a special case in which the bound from Theorem \ref{thm:prodvar} can be improved. We say that a product state $|v\rangle$ is locally optimal for $H$ if, for any single-qubit Pauli $Q$ we have
\[
\frac{d}{d\phi} \langle v|e^{-i\phi Q} H e^{i\phi Q}|v\rangle \big|_{\phi=0} =0,
\]
or equivalently
\begin{equation}
\langle v|[Q,H]|v\rangle=0.
\label{eq:qcom}
\end{equation}
As in the above, for simplicity we shall work in a local basis defined by $v$, so that $|v\rangle=|0^n\rangle$.

\begin{claim}
Suppose $|0^n\rangle$ is locally optimal for $H$. Then for any string $z\in \{0,1\}^n$ with Hamming weight $|z|=1$ we have
\begin{equation}
\langle z|H|0^n\rangle=0.
\label{eq:zh}
\end{equation}
\label{claim:no1}
\end{claim}
\begin{proof}
Without loss of generality consider the case where $z=10^{n-1}$. Then
\[
|2\mathrm{Im}(\langle z|H|0^n\rangle)|=|\langle0^n |[X_1,H]|0^n\rangle|=0 \quad \text{and} \quad |2\mathrm{Re}(\langle z|H|0^n\rangle)|=|\langle0^n |[Y_1,H]|0^n\rangle|=0,
\]
where we used Eq.~\eqref{eq:qcom}. 
\end{proof}

\begin{claim}
Suppose $|0^n\rangle$ is locally optimal for $H$. We may efficiently choose $\{P_i\}$ and $\{\theta_{ij}\}$ such that
\begin{equation}
\langle \psi|H|\psi\rangle\geq \langle 0^n|H|0^n\rangle+ \Omega\left(\frac{\mathrm{Var}(H)^2}{d|E|}\right).
\label{eq:alphabnd}
\end{equation}
\end{claim}
\begin{proof}
Since $H$ is two-local we have
\[
\mathrm{Var}(H)=\langle0^n|HQ_1H|0^n\rangle+\langle 0^n|HQ_2H|0^n\rangle=\langle 0^n|HQ_2H|0^n\rangle,
\]
where in the last equality we used claim \ref{claim:no1}. The claim then follows directly by combining Proposition \ref{prop:q2} and Theorem~\ref{thm:performance}.
\end{proof}

\section{Improvement of random states}
We prove the first part of Theorem \ref{thm:random} regarding general degree-$d$ graphs, which is implied by the following lemma. 
\begin{lemma}
\label{lem:randprod}
Let $\ket{v}=|v_1\rangle\otimes|v_2\rangle\ldots \otimes |v_n\rangle$ where each $v_i$ is a Haar random single-qubit state. Then there is an efficient randomized process with random coins $r$, that constructs the matrices $P_i$ (depending on both $r$ and $\ket{v}$) such that the resulting state $\ket{\psi_{r,v}}$ satisfies
$$\expec_{r,v} \bra{\psi_{r,v}}H\ket{\psi_{r,v}} \geq \expec_v \bra{v}H\ket{v} + \Omega\left(\frac{\mathrm{quad}(H)^2}{d|E|}\right).$$
\end{lemma}

\begin{proof}
Pick $\ket{v}=\otimes_i\ket{v_i}$, where each $\ket{v_i}$ is chosen uniformly at random from Haar measure on qubits. Also choose $n$ uniformly random real numbers $\mu_i$ i.i.d in the interval $[0,\frac{\pi}{2}]$. The latter choice is made using the coins $r$. Given $\ket{v},r$  define $P_i= e^{i \mu_i}\ket{v_i}\bra{v_i^{\perp}}+e^{-i \mu_i}\ket{v^{\perp}_i}\bra{v_i}$ (we drop the labels $\ket{v},r$ from $P_i$ for convenience). Observe that $\bra{v_i}P_i\ket{v_i}=0$ and $\|P_i\|\leq 1$, as required. Then $\alpha_{v,r}$ (as given in Eq~\eqref{eq:alpha} ) can be evaluated to be 
\beqar
\label{eq:alphavrdef}
\alpha_{v,r}&=&\expec_{\{i,j\}\in E}\left|\br{e^{i(\mu_i+\mu_j)}\bra{v_i^{\perp},v_j^{\perp}}h_{ij}\ket{v_i,v_j} -e^{-i(\mu_i+\mu_j)}\bra{v_i,v_j}h_{ij}\ket{v^{\perp}_i,v^{\perp}_j}}\right|\nonumber\\
&=& 2\expec_{\{i,j\}\in E}\left|\text{Im}\br{e^{i(\mu_i+\mu_j)}\bra{v_i^{\perp},v_j^{\perp}}h_{ij}\ket{v_i,v_j}}\right|.
\enqar
Let $\bra{v_i^{\perp},v_j^{\perp}}h_{ij}\ket{v_i,v_j}= e^{i \kappa_{i,j}}|\bra{v_i^{\perp},v_j^{\perp}}h_{ij}\ket{v_i,v_j}|$ be the polar decomposition. Then  
$$|\text{Im}\br{e^{i(\mu_i+\mu_j)}\bra{v_i^{\perp},v_j^{\perp}}h_{ij}\ket{v_i,v_j}}|=|\sin\br{\mu_i+\mu_j+\kappa_{i,j}}|\cdot |\bra{v_i^{\perp},v_j^{\perp}}h_{ij}\ket{v_i,v_j}|.$$
Note that
\beqarst
\expec_r|\sin\br{\mu_i+\mu_j+\kappa_{i,j}}|= \frac{4}{\pi^2}\int_{0}^{\frac{\pi}{2}}\int_{0}^{\frac{\pi}{2}}|\sin\br{\mu_i+\mu_j+\kappa_{i,j}}|d\mu_id\mu_j\geq \frac{2}{5},
\enqarst
for all $\kappa_{i,j}$. Then Eq \ref{eq:alphavrdef} ensures that
$$\expec_r\alpha_{v,r}= 2\expec_{\{i,j\}\in E}\expec_r|\text{Im}\br{e^{i(\mu_i+\mu_j)}\bra{v_i^{\perp},v_j^{\perp}}h_{ij}\ket{v_i,v_j}}|\geq \frac{4}{5}\cdot\expec_{\{i,j\}\in E}|\bra{v_i^{\perp},v_j^{\perp}}h_{ij}\ket{v_i,v_j}|.$$
 Then we can evaluate 
\beqar
\label{eq:haarval}
&&\expec_{v,r} \alpha_{v,r}\geq\frac{4}{5}\cdot\expec_{\{i,j\}\in E}\int |\br{\bra{v_i^{\perp},v_j^{\perp}}h_{ij}\ket{v_i,v_j}}|dv_i dv_j\nonumber\\
&&\geq \frac{4}{5}\cdot\expec_{\{i,j\}\in E}\int |\bra{v_i^{\perp},v_j^{\perp}}h_{ij}\ket{v_i,v_j}|^2dv_i dv_j\nonumber\\
&&= \frac{4}{5}\cdot\expec_{\{i,j\}\in E}\int \tr\br{\ketbra{v_i^{\perp},v_j^{\perp}}{v_i^{\perp},v_j^{\perp}}h_{ij}\ketbra{v_i,v_j}{v_i,v_j}h_{ij}}dv_i dv_j\nonumber\\
&&=\frac{4}{5}\cdot\expec_{\{i,j\}\in E}\int \bra{11}\br{U_i^{\dagger}\otimes V_j^{\dagger}}h_{ij}\br{U_i\otimes V_j}\ketbra{00}{00}\br{U_i^{\dagger}\otimes V_j^{\dagger}}h_{ij}\br{U_i\otimes V_j}\ket{11}dU_i dV_j\nonumber\\
&&=\frac{4}{5}\cdot\expec_{\{i,j\}\in E}\int \bra{1100}\br{U_{i1}^{\dagger}\otimes V_{j1}^{\dagger}\otimes U_{i2}^{\dagger}\otimes V_{j2}^{\dagger}}h_{i1,j1}\otimes h_{i2,j2}\br{U_{i1}\otimes V_{j1}\otimes U_{i2}\otimes V_{j2}}\ket{0011}dU_i dV_j,\nonumber\\
\enqar
where in the second last equality we fixed a basis $\{\ket{0},\ket{1}\}$ for each qubit and introduced random unitaries $U_i, V_j$ that specify $\ket{v_i}=U_i\ket{0}, \ket{v_j}=V_j\ket{0}$. Using the well known properties of Haar integral, we have
\beqar
\label{eq:haarint}
&&\int \br{U_{i1}^{\dagger}\otimes V_{j1}^{\dagger}\otimes U_{i2}^{\dagger}\otimes V_{j2}^{\dagger}}h_{i1,j1}\otimes h_{i2,j2}\br{U_{i1}\otimes V_{j1}\otimes U_{i2}\otimes V_{j2}}dU_i dV_j\nonumber\\
&&= a \id_{i1,i2}\otimes\id_{j1,j2} + b S_{i1,i2}\otimes\id_{j1,j2} + c \id_{i1,i2}\otimes S_{j1,j2} + d S_{i1,i2}\otimes S_{j1,j2}.
\enqar
Above, $\id$ is the identity operator, $S$ is the swap operator and the subscripts represent the qubits on which the operator acts. Coefficients $a,b,c,d$ can be evaluated using the following system of equations, obtained from Eq. \ref{eq:haarint} by tracing each of the four operators.
\beqarst
&&\br{\tr_{i,j}{h_{ij}}}^2= 16a + 8b + 8c + 4d\\
&&\tr_j\br{\tr_i{h_{ij}}\tr_i{h_{ij}}}= 8a+4b+16c+8d\\
&&\tr_i\br{\tr_j{h_{ij}}\tr_j{h_{ij}}}= 8a+ 16b + 4c+ 8d\\
&& \tr_{i,j}\br{h_{ij}^2}= 4a+8b+8c+16d.
\enqarst
One can solve for $d$ to obtain 
$$d=\frac{\br{\tr_{i,j}{h_{ij}}}^2}{36}+\frac{\tr_{i,j}\br{h_{ij}^2}}{9}-\frac{\tr_j\br{\tr_i{h_{ij}}\tr_i{h_{ij}}}+\tr_i\br{\tr_j{h_{ij}}\tr_j{h_{ij}}}}{18}.$$
In order to obtain a simpler lower bound and see that $d$ is positive, we expand $h_{ij}= \sum_{x,y} f^{i,j}_{x,y}\sigma^i_x\otimes\sigma^j_y$ in the two qubit Pauli basis. Then 
$$\br{\tr_{i,j}h_{ij}}^2=16\br{f^{i,j}_{0,0}}^2, \quad \tr_{i,j}\br{h_{ij}^2}=4\sum_{x,y}\br{f^{i,j}_{x,y}}^2,$$ $$\tr_j\br{\tr_i{h_{ij}}\tr_i{h_{ij}}}=8\sum_{y}\br{f^{i,j}_{0,y}}^2, \quad \tr_i\br{\tr_j{h_{ij}}\tr_j{h_{ij}}}=8\sum_{y}\br{f^{i,j}_{y,0}}^2.$$
Hence,
\beqarst
&&\frac{\br{\tr_{i,j}{h_{ij}}}^2}{36}+\frac{\tr_{i,j}\br{h_{ij}^2}}{9}-\frac{\tr_j\br{\tr_i{h_{ij}}\tr_i{h_{ij}}}+\tr_i\br{\tr_j{h_{ij}}\tr_j{h_{ij}}}}{18}\\
&&=\frac{4}{9}\br{\br{f^{i,j}_{0,0}}^2+\sum_{x,y}\br{f^{i,j}_{x,y}}^2-\sum_{y}\br{f^{i,j}_{0,y}}^2-\sum_{y}\br{f^{i,j}_{y,0}}^2}\\
&&=\frac{4}{9}\br{\sum_{x>0,y>0}\br{f^{i,j}_{x,y}}^2}.
\enqarst
Conjugating Eq. \ref{eq:haarint} with $\bra{1100}\br{\cdot}\ket{0011}$, it can be seen that only the term corresponding to $d$ survives and evaluates to $1$. Thus, Eq. \ref{eq:haarval} gives
\beqarst
&&\expec_{v,r} \alpha_{v,r}\geq \\
&& \frac{4}{5}\expec_{\{i,j\}\in E}\br{\frac{\br{\tr_{i,j}{h_{ij}}}^2}{36}+\frac{\tr_{i,j}\br{h_{ij}^2}}{9}-\frac{\tr_j\br{\tr_i{h_{ij}}\tr_i{h_{ij}}}+\tr_i\br{\tr_j{h_{ij}}\tr_j{h_{ij}}}}{18}}\\
&& = \frac{16}{45}\expec_{\{i,j\}\in E}\br{\sum_{x>0,y>0}\br{f^{i,j}_{x,y}}^2}=\frac{16}{45}\frac{\mathrm{quad}(H)}{|E|}.
\enqarst
Thus, using the convexity of square function,
$$\expec_{v,r} \alpha^2_{v,r}\geq \frac{16^2}{45^2}\br{\frac{\mathrm{quad}(H)}{|E|}}^2\geq \frac{1}{8}\br{\frac{\mathrm{quad}(H)}{|E|}}^2.$$
This completes the proof by employing Theorem \ref{thm:general bound}.
\end{proof}

\subsection{Triangle-free graphs}

In this section we establish the second part of Theorem \ref{thm:random}, which concerns triangle-free graphs. The proof is based on the following exact expression. It will be convenient in what follows to work in a local basis in which the product state of interest is $|v\rangle=|0^n\rangle$.

\begin{lemma}[\textbf{Improvement for triangle-free Hamiltonians}]
Suppose $G$ is a triangle-free, degree-$d$ graph. Suppose we are given single-qubit Hermitian operators $\{P_i\}_{i\in [n]}$ satisfying $P_i^2=I$ and $\langle 0|P_i|0\rangle=0$ for all $i\in [n]$, and consider the state $|\psi\rangle=e^{i\sum_{\{r,s\}\in E} \theta_{rs} P_r P_s}|0^n\rangle$ as a function of the real parameters $\{\theta_{rs}\}$. Define
\[
\alpha_{kl}=|\langle 00|[h_{kl}, P_kP_l]|00\rangle|
\]
We can efficiently choose $\theta_{ij}\in \{\pm \theta\}$ for each edge $\{i,j\}\in E$ so that
\begin{align}
\langle \psi| h_{kl}|\psi\rangle&= \frac{1}{4}\mathrm{Tr}(h_{kl})+ \frac{1}{4}\mathrm{Tr}(h_{kl} Z_k Z_l) \cos^{2d-2}(2\theta)+\frac{1}{4}\mathrm{Tr}(h_{kl}(Z_k+Z_l)) \cos^{d}(2\theta).+\frac{\alpha_{kl}}{2} \sin(2\theta)\cos^{d-1}(2\theta)
\label{eq:psien1}
\end{align}
for all edges $\{k,l\}\in E$
\end{lemma}

\begin{proof}
We have
\begin{equation}
\langle \psi| h_{kl}|\psi\rangle = \langle 0^n|V^{\dagger}_{kl} h_{kl}(\theta) V_{kl}|0^n\rangle
\label{eq:psie}
\end{equation}
where $h_{kl}(\theta_{kl})=e^{-i\theta_{kl} P_k P_l} h_{kl}e^{i\theta_{kl} P_k P_l}$ and
\begin{align}
V_{kl}& =\prod_{\{k,s\}\in E\setminus \{k,l\}} e^{i\theta_{ks} P_k P_s} \prod_{\{r,l\}\in E\setminus \{k,l\}} e^{i\theta_{rl} P_r P_l}\\
&=\prod_{\{k,s\}\in E\setminus \{k,l\}} \left(\cos(\theta)+i\sin(\theta_{ks})P_kP_s\right)\prod_{\{r,l\}\in E\setminus \{k,l\}}  \left(\cos(\theta)+i\sin(\theta_{rl})P_rP_l\right).
\label{eq:vkl}
\end{align}
Plugging Eq.~\eqref{eq:vkl} into Eq.~\eqref{eq:psie} and using $\langle 0|P_i|0\rangle=0$ and the fact that $G$ is triangle-free gives
\begin{align*}
\langle \psi| h_{kl}|\psi\rangle=\sum_{A\subseteq N(k)\setminus{\{l\}}} \sum_{B\subseteq N(l)\setminus{\{k\}}}& \left(\cos^2(\theta)\right)^{2d-2-|A|-|B|} \left(\sin^2(\theta)\right)^{|A|+|B|} \\
& \cdot \langle 0^n| \left(\prod_{s\in A} P_k P_s \prod_{r\in B} P_r P_l\right) h_{kl}(\theta_{kl}) \left(\prod_{s\in A} P_k P_s \prod_{r\in B} P_r P_l\right)|0^n\rangle
\end{align*}
In the above we also used our choice $|\theta_{ij}|=\theta$ for all edges $\{i,j\}\in E$. Observe that the matrix element appearing in the above depends only on the parity (even/odd) of $|A|$ and $|B|$. In particular,
\[
\langle \psi| h_{kl}|\psi\rangle=F_{EE}+F_{EO}+F_{OE}+F_{OO} 
\]
where the even/even term is
\begin{align}
F_{EE}&=\langle 00|h_{kl}(\theta_{kl})|00\rangle\left(\sum_{j=0,2,\ldots} {{d-1}\choose {j} }\left(\cos^2(\theta)\right)^{d-1-j} \left(\sin^2(\theta)\right)^{j}\right)^2\\
&=\langle 00|h_{kl}(\theta_{kl})|00\rangle\frac{1}{4}\left(1+\cos^{d-1}(2\theta)\right)^2,
\label{eq:fee}
\end{align}
and by similar calculations one arrives at
\begin{align}
F_{EO}& =\langle 10| h_{kl}(\theta_{kl})|10\rangle  \frac{1}{4}\left(1-\cos^{2d-2}(2\theta)\right)\label{eq:feo}\\
F_{OE}& =\langle 01| h_{kl}(\theta_{kl})|01\rangle  \frac{1}{4}\left(1-\cos^{2d-2}(2\theta)\right) \label{eq:foe}\\
F_{OO} &=\langle 11|h_{kl}(\theta_{kl})|11\rangle\frac{1}{4}\left(1-\cos^{d-1}(2\theta)\right)^2 
\label{eq:foo}
\end{align}

Now for ease of presentation in the following we write $c=\cos^{d-1}(2\theta)$ and $a_{xy}=\langle x y|h_{kl}(\theta_{kl})|x y\rangle$, for $x,y\in \{0,1\}$. Then expanding the above expression gives
\begin{align}
\langle \psi| h_{kl}|\psi\rangle &=F_{EE}+F_{EO}+F_{OE}+F_{OO} \label{eq:feq}\\
&= a_{00}+(1-c)\left( \frac{a_{01}+a_{10}}{2}-a_{00}\right)+ \frac{1}{4}(1-c)^2 \left( a_{11} +a_{00}-a_{01}-a_{10}\right).\label{eq:ca}
\end{align}

Now let 
\begin{align}
    b_{xy}=\langle xy|h_{kl}|xy\rangle.
\end{align}
So that
\begin{align*}
a_{00}=\cos^2 (\theta)b_{00}+\sin^2(\theta) b_{11}+i\cos(\theta) \sin(\theta_{kl})\langle 00|[h_{kl}, P_kP_l]|00\rangle.
\end{align*}
We now fix the sign of $\theta_{kl}$ so that 
\begin{equation}
a_{00}=\cos^2 (\theta)b_{00}+\sin^2(\theta) b_{11}+\frac{1}{2}\sin(2\theta)\alpha_{kl}.
\label{eq:00}
\end{equation}
With this choice we have
\begin{equation}
a_{11}=\cos^2 (\theta)b_{11}+\sin^2(\theta) b_{00}-\frac{1}{2}\sin(2\theta)\alpha_{kl}.
\label{eq:11}
\end{equation}
To compute the third term in the above equation we used the fact that $P_k^2=P_l^2=I$ and $\langle0|P_k|0\rangle=\langle 0|P_l|0\rangle=0$ which implies
\[
\langle 11|[P_k P_l, h_{kl}]|11\rangle =\langle 00|P_k P_l[P_k P_l, h_{kl}]P_k P_l|00\rangle=-\langle 00|[P_k P_l, h_{kl}]|00\rangle.
\]
Similarly, by a direct calculation we get
\begin{equation}
a_{01}+a_{10}=b_{10}+b_{01}.
\label{eq:10}
\end{equation}

Plugging Eqs.~(\ref{eq:00}, \ref{eq:11}, \ref{eq:10}) into Eq.~\eqref{eq:ca} we get
\begin{align}
\langle \psi| h_{kl}|\psi\rangle &= b_{00} + (a_{00}-b_{00})c +(1-c)\left(\frac{b_{01}+b_{10}}{2}-b_{00}\right)+\frac{1}{4}(1-c)^2\left(b_{00}+b_{11}-b_{10}-b_{01}\right)\nonumber\\
&=b_{00}+\frac{\alpha_{kl}}{2}\sin(2\theta)\cos^{d-1}(2\theta)+\sin^2(\theta)\cos^{d-1}(2\theta)\left(b_{11}-b_{00}\right)\nonumber\\
&+(1-\cos^{d-1}(2\theta)) \left(\frac{b_{01}+b_{10}}{2}-b_{00}\right)+\frac{1}{4}(1-\cos^{d-1}(2\theta))^2\left(b_{00}+b_{11}-b_{10}-b_{01}\right)
\end{align}
Rearranging the above expression we arrive at

\begin{align}
\langle \psi| h_{kl}|\psi\rangle&= \frac{\alpha_{kl}}{2}\sin(2\theta)\cos^{d-1}(2\theta)+b_{00}\left(\frac{1}{4}+\frac{1}{4}\cos^{2d-2}(2\theta)+\frac{1}{2}\cos^{d}(2\theta)\right)\nonumber\\
&+b_{11}\left(\frac{1}{4}+\frac{1}{4}\cos^{2d-2}(2\theta)-\frac{1}{2}\cos^{d}(2\theta)\right)+\left(b_{01}+b_{10}\right)\left(\frac{1}{4}-\frac{1}{4}\cos^{2d-2}(2\theta)\right)
\label{eq:trifree}
\end{align}

By noting that $\sum_{x,y} b_{xy} = \mathrm{Tr}(h_{kl})$, $\sum_{x,y} (-1)^{x+y} b_{xy} = \mathrm{Tr}(h_{kl}Z_kZ_l)$, and $b_{00}-b_{11} = \dfrac{1}{2}\mathrm{Tr}(h_{kl}(Z_k+Z_l))$, we arrive at Eq.~\eqref{eq:psien1}.

\end{proof}

Using the expression in \eqref{eq:psien1}, we prove the bound for triangle-free graphs from Theorem \ref{thm:random}:

\begin{proof}
As shown above, the exact formula for the energy of $|\psi\rangle=V(\vec{\theta})|0^n\rangle$ on a triangle-free graph is
\begin{align}
\langle \psi| h_{kl}|\psi\rangle&= \frac{1}{4}\mathrm{Tr}(h_{kl})+ \frac{1}{4}\mathrm{Tr}(h_{kl} Z_k Z_l) \cos^{2d-2}(2\theta)\nonumber\\
&+\frac{1}{4}\mathrm{Tr}(h_{kl}(Z_k+Z_l)) \cos^{d}(2\theta).+\frac{\alpha_{kl}}{2} \sin(2\theta)\cos^{d-1}(2\theta).
\label{eq:psien}
\end{align}
Here $\alpha_{kl}$ depends on the choices of $P_k, P_l$. We either choose $P_i=X_i$ for all $i$, or $P_i=(X+Y)_i/\sqrt{2}$ for all $i$, each with probability $1/2$. Then
\[
\mathbb{E}(\alpha_{kl})=2|\mathrm{Re}(\langle 00|h_{kl}|11\rangle)|+2|\mathrm{Im}(\langle 00|h_{kl}|11\rangle)|\geq 2|\langle 00|h_{kl}|11\rangle|.
\]
Substituting in Eq.~\eqref{eq:psien} gives
\begin{align}
\mathbb{E}\left(\langle \psi| h_{kl}|\psi\rangle\right)&\geq \frac{1}{4}\mathrm{Tr}(h_{kl})+ \frac{1}{4}\mathrm{Tr}(h_{kl} Z_k Z_l) \cos^{2d-2}(2\theta)\nonumber\\
&+\frac{1}{4}\mathrm{Tr}(h_{kl}(Z_k+Z_l)) \cos^{d}(2\theta).+|\langle 00|h_{kl}|11\rangle| \sin(2\theta)\cos^{d-1}(2\theta).
\label{eq:psien2}
\end{align}
Now instead of using the starting state $|0^n\rangle$, suppose we start from a random computational basis state $|s\rangle=X(s)|0^n\rangle$. Running through the above argument in the rotated basis defined by $|s\rangle$ we see that for a suitable random choice of $\{P_i\}$ we have
\begin{align}
\mathbb{E}\left(\langle \psi| h_{kl}|\psi\rangle\right)&\geq \frac{1}{4}\left(\mathrm{Tr}(h_{kl})\right)+\frac{|\langle 00|h_{kl}|11\rangle|+|\langle 01|h_{kl}|10\rangle|}{2} \sin(2\theta)\cos^{d-1}(2\theta)\nonumber\\
&\geq \frac{1}{4}\left(\mathrm{Tr}(h_{kl})\right)+\frac{1}{4}|\mathrm{Tr}(h_{kl} X_kX_l)|\sin(2\theta)\cos^{d-1}(2\theta).
\label{eq:xx}
\end{align}
Here we used the fact that 
\[
\mathbb{E}_s\left(\mathrm{Tr}(X(s)h_{kl}X(s) Z_k Z_l)\right)=\mathbb{E}_s\left(\mathrm{Tr}(X(s)h_{kl}X(s)(Z_k+Z_l)) \right)=0.
\]
Summing Eq.~\eqref{eq:xx} over all edges $\{k,l\}\in E$ gives
\[
\mathbb{E}\left(\langle \psi| H|\psi\rangle\right)\geq \frac{1}{4}\left(\mathrm{Tr}(H)\right)+\frac{1}{4}\sin(2\theta)\cos^{d-1}(2\theta) \sum_{\{k,l\}\in E} |\mathrm{Tr}\left(h_{kl} X_k X_l\right)|
\]
Since there is nothing special about the $X$-basis we can again use our freedom to randomize the local basis of each qubit to get
\begin{align}
\mathbb{E}\left(\langle \psi| H|\psi\rangle\right)&\geq \frac{1}{4}\mathrm{Tr}(H)+\frac{1}{36}\sin(2\theta)\cos^{d-1}(2\theta) \sum_{\{k,l\}\in E}\;  \sum_{Q,R\in \{X,Y,Z\}} |\mathrm{Tr}\left(h_{kl} Q_k R_l\right)|\nonumber\\
&\geq  \frac{1}{4}\mathrm{Tr}(H)+\frac{1}{36}\sin(2\theta)\cos^{d-1}(2\theta) \sum_{\{k,l\}\in E}\;  \sum_{Q,R\in \{X,Y,Z\}} |\mathrm{Tr}\left(h_{kl} Q_k R_l\right)|^2/4\nonumber\\
&=\frac{1}{4}\mathrm{Tr}(H)+\sin(2\theta)\cos^{d-1}(2\theta) \frac{\mathrm{quad}(H)}{36},
\end{align}
where in the second-to-last line we used the fact that $|\mathrm{Tr}\left(h_{kl} Q_k R_l\right)|\leq 4$ which follows from $\|h_{kl}\|\leq 1$. Finally, we can find the maximum value of the second term with respect to $\theta$ by noting that $\sin(2\theta)\cos^{d-1}(2\theta)$ reaches a maximum when $\theta=\arcsin(\frac{1}{\sqrt{d}})$. Using this fact we get

\begin{equation}
    \mathbb{E}(\langle \psi |H|\psi\rangle)\geq \frac{1}{4}\mathrm{Tr}(H)+\Omega\left(\frac{\mathrm{quad}(H)}{\sqrt{d}}\right)\label{eq:rr35}.
\end{equation}
\end{proof}

\section{Improvement of bounded-depth states}
We prove Theorem \ref{thm:lowdepthimp}. Given the $d$-regular graph $G=(V,E)$, we consider the state $\ket{v}=W\ket{0}^n$, where $W$ has a maximum lightcone of size $\ell$. The aim is to increase the energy of $\ket{v}$ with respect to $H$.  The light cones of the edges have sizes at most $2\ell$. Define 
\beq
\label{eq:Gdef}
F= \sum_{j=1}^n W\ketbra{1}{1}_j W^{\dagger}.
\enq
 The locality of $F$ is $\ell$. Let $A= i[H, F]$ and define $\ket{\psi}= e^{iA\theta}\ket{v}$ (thus $U=e^{iA\theta}$ in the statement of Theorem \ref{thm:lowdepthimp}). We can write 
\beq
\label{eq:Wdef}
A=\sum_{e\in E} i[h_e,F]:=\sum_{e\in E} A_e,
\enq
where 
\beqar
A_e= i\sum_{j=1}^n[h_e, W\ketbra{1}{1}_j W^{\dagger}]=i\sum_{j: \text{supp}(h_e) \cap \text{supp}(W\ketbra{1}{1}_j W^{\dagger})\neq \phi}[h_e, W\ketbra{1}{1}_j W^{\dagger}]. 
\enqar
 Any $j$ satisfying $\text{supp}(h_e) \cap \text{supp}(W\ketbra{1}{1}_j W^{\dagger})\neq \phi$ is in the light cone of $h_e$. Thus there are $\leq 2\ell$ such $j$'s. For any such $j$, $W\ketbra{1}{1}_j W^{\dagger}$ has locality $\ell$. Thus, $A_e$ is supported on $\leq 2\ell^2$ qubits. Further, $$\|A_e\|\leq 2\ell\cdot \max_{j: \text{supp}(h_e) \subset \text{supp}(W\ketbra{1}{1}_j W^{\dagger})}\|[h_e, W\ketbra{1}{1}_j W^{\dagger}]\| \leq 2\ell,$$ where we used $$\|[h_e, W\ketbra{1}{1}_j W^{\dagger}]\|=\|[h_e, W\ketbra{1}{1}_j W^{\dagger}-\id/2]\|\leq 2\|h_e\|\|W\ketbra{1}{1}_j W^{\dagger}-I/2\|\leq 1.$$ We have
\beq
\label{eq:energy}
\bra{\psi}H\ket{\psi}=\bra{v}H\ket{v}- i\bra{v}[A,H]\ket{v}\theta+\sum_{m=2}^\infty \frac{(-i\theta)^m}{m!}\bra{v}[A,H]_m\ket{v},
\enq
 Now, using the identities $F\ket{v}=0$ and $F \geq \mathrm{I}-\ketbra{v}{v}$, we find
\beq
\label{eq:firstorder}
-i\bra{v}[A,H]\ket{v}\theta=\bra{v}[[H, F], H]\ket{v}\theta = 2\bra{v}HFH\ket{v}\theta\geq 2\theta\bra{v}H(\mathrm{I}-\ketbra{v}{v})H\ket{v}= 2\theta \mathrm{Var}(H).
\enq
Thus, let us focus on the terms with $m\geq 2$. We upper bound 
\beq
\label{eq:mupbound}
\bra{v}[A,H]_m\ket{v}\leq \sum_{e\in E}\|[A,h_e]_m\|\leq |E|\max_e\|[A, h_e]_m\|.
\enq
Now, consider for each $e$,
\beq
\label{eq:commexpand}
[A,h_e]_m= \sum_{e_1,\ldots e_m}[A_{e_m},[A_{e_{m-1}}\ldots[A_{e_1},h_e]]],
\enq
where we used Eq. \ref{eq:Wdef}. Most terms are zero and we will bound the number of non-zero terms. We will use the following simple fact.
\begin{fact}
\label{fact:subset}
Let $S\subset V$. The number of $e$ such that the support of $A_e$ overlaps with $S$ is at most $|S|\ell^2 d$. For each such $e$ and any operator $O_S$ on $S$ , the support of $[O_S, A_e]$ is at most $|S|+2\ell^2$.
\end{fact}
\begin{proof}
Since $e$ is such that the support of $A_e$ overlaps with $S$, there exist a $j$ satisfying $\text{supp}(h_e) \cap \text{supp}(W\ketbra{1}{1}_j W^{\dagger})\neq \phi$ for which the support of $[h_e,W\ketbra{1}{1}_j W^{\dagger}]$ overlaps with $S$. Thus, either $\text{supp}(h_e)\cap S \neq \phi$ or $j$ belongs to the light cone of $S$. Since the support of $h_e$ overlaps with the light cone of $j$ in the latter case, we have that the support of $h_e$ overlaps with the light cone of the light cone of $S$ (in both the cases). To upper bound the number of possible $e$, we hence we count the size of the light cone of the light cone of $S$ ($\leq|S|\ell^2$)  and then count the number of edges intersecting with this light cone ($\leq d|S|\ell^2$). For any such $e$, the support of $[O_S, A_e]$ is contained in the union of $S$ and the support of $A_e$. This completes the proof. 
\end{proof}
Using Fact \ref{fact:subset}, let us estimate the number of $(e_1,e_2,\ldots e_m)$ that contribute to Eq. \ref{eq:commexpand}. Setting $S$ to be the set of two vertices of $e$, we find that the number of $e_1$ is at most $2d\ell^2$. Arguing inductively, suppose we have fixed $e_1,e_2,\ldots e_{k-1}$. The support size of $[A_{e_{k-1}},[A_{e_{m-1}}\ldots[A_{e_1},h_e]]]$ is at most $2+2(k-1)\ell^2$ (by Fact \ref{fact:subset}). Thus, the number of $e_k$ contributing to Eq. \ref{eq:commexpand} is at most $(2+2(k-1)\ell^2)\ell^2 d$. Hence, the total number of $(e_1,\ldots e_m)$ is at most
$$(2d\ell^2)\cdot (2+2\ell^2)\ell^2 d \cdot (2+4\ell^2)\ell^2 d \ldots (2+2(m-1)\ell^2)\ell^2 d \leq 2^m\cdot (m-1)!\cdot \ell^{2m-2}\cdot (2d\ell^2)^m \leq (m-1)! (4d\ell^4)^m.$$
Thus,
\beqar
\|[A,h_e]_m\|_\infty&\leq& (m-1)! (4d\ell^4)^m \max_{e_1,\ldots e_m}\|[A_{e_m},[A_{e_{m-1}}\ldots[A_{e_1},h_e]]]\|\nonumber\\
&\leq& (m-1)! (4d\ell^4)^m \cdot 2^m\|h_e\|\cdot\max_{e_1,\ldots e_m}\|A_{e_1}\|\cdot \|A_{e_2}\|\ldots \cdot \|A_{e_m}\|\nonumber\\
&\overset{(a)}\leq& (m-1)! (4d\ell^4)^m\cdot 2^m \cdot (2\ell)^m = (m-1)!\cdot (16d\ell^5)^m,
\enqar  
where $(a)$ uses $\|A_e\|\leq 2\ell$. Combining with Eq. \ref{eq:mupbound}, this ensures that
\beqar
\sum_{m=2}^\infty \left|\frac{(i\theta)^m}{m!}\bra{v}[[H,A]]_m\ket{v}\right|&\leq& |E|\sum_{m=2}^\infty \frac{\theta^m}{m!} (m-1)!\cdot (16d\ell^5)^m\nonumber\\
&\leq& |E|\cdot \sum_{m=2}^\infty (16d\ell^5\theta)^m\nonumber\\
&\leq& 2|E|\cdot (16d\ell^5\theta)^2, 
\enqar
where the last inequality assumes $\theta \leq\frac{1}{32d\ell^5}$ (our choice below will satisfy this). Thus, using Eq. \ref{eq:energy} and Eq. \ref{eq:firstorder},
\beqar
\bra{\theta}H\ket{\theta}&\geq&\bra{v}H\ket{v} + 2\theta \mathrm{Var}(H) - 2|E|\cdot (16d\ell^5\theta)^2\nonumber\\
&=& \bra{v}H\ket{v} + 2\theta |E|\br{\frac{\mathrm{Var}(H)}{|E|} - 2(16d\ell^5)^2\theta}.
\enqar
Setting $\theta= \frac{\mathrm{Var}(H)}{2^{10}d^2\ell^{10}|E|} \leq \frac{1}{32d\ell^5}$, we conclude that
\beq
\bra{\psi}H\ket{\psi}\geq \bra{v}H\ket{v} + \frac{\mathrm{Var}(H)^2}{2^{10}d^2\ell^{10}|E|}.
\enq

We highlight that the above proof can be applied with minor modifications to the more general case in which $F$ is a $\ell$-local Hamiltonian with the unique ground state $\ket{v}$ and constant spectral gap $\gamma=\Omega(1)$. In this case, we set the ground energy of $F$ at $0$, leading to the relations  $F\ket{v}=0$ and $F\succeq \gamma \left(\id - \ketbra{v}{v}\right)$. Thus, the first order contribution in \eqref{eq:firstorder} is replaced by $$-i\bra{v}[A,H]\ket{v}\theta\geq 2\gamma\theta \mathrm{Var}(H).$$ The higher order contributions are upper bounded in a manner similar to above.

\section{Improvement for general $k$-local Hamiltonians}

Let $G=(V,E)$ be a hypergraph with hyperedges of size at most $k$ and $n=|V|$ qubits on its vertices. We denote the number of hyperedges that contain $i\in V$ by $\deg(i)$ and assume $\deg(i)\leq d$ for all $i \in V$. Consider a $k$-local Hamiltonian $H=\sum_{R\in E} h_R$ where each local term $h_R$ acts non-trivially on a subset $R\in E$ of qubits with $|R|\leq k$ and $\norm{h_R}\leq 1$. Here without loss of generality we assume the input product states is $\ket{v}=\ket{0^n}$. We use a similar argument as in the proof of Theorem~\ref{thm:prodvar} to relate the improvement after applying an extension of the quantum circuit $V(\theta)$ in Eq.~\eqref{eq:rr2} to the variance $\Var(H)$. To this end, we write
\[
\Var(H)=\sum_{t=1}^{k} \bra{0^n} H Q_t H\ket{0^n}
\]
where $Q_t$ are the projector onto the computational basis states with Hamming weight $t$. Note that operators $Q_t$ with Hamming weight $>k$ do not contribute. It holds that there exists a $t$ such that 
\[\langle 0^n| H Q_t H|0^n\rangle \geq \frac{1}{k} \text{Var}(H).\]

Depending on $t$, our choice of circuit $V(\theta)$ is a generalization of what we had before in Theorem~\ref{thm:prodvar}. The set $S$ contains all the collection of $t$ different vertices $\{j_1,j_2,\dots,j_t\}$ which fully reside in the support of at least one local term $h_R$ of the Hamiltonian $H$. That is, there exists an $R$ such that $\{j_1,j_2,\dots,j_t\}\subseteq \supp(h_R)$.  Define $V(\theta)=e^{-i\theta L}$ where
$$L=\sum_{\{j_1,j_2,\dots,j_t\}\in S}(-1)^{a_{j_1,\dots,j_t}} P_{j_1}P_{j_2}\dots P_{j_t}.$$
Here each $P_j$ is a single-qubit Pauli operator with the property $\bra{0}P_j\ket{0}=0$ that acts nontrivially only on qubit $j$ and $a_{j_1,\dots,j_t}\in \{0,1\}$ is chosen so that 
\[
i\langle 0^n|(-1)^{a_{j_1,\dots,j_t}} [P_{j_1}P_{j_2}\dots P_{j_t},H]|0^n\rangle = |\langle 0^n| [P_{j_1}P_{j_2}\dots P_{j_t},H]|0^n\rangle|.
\] 
Let $|\theta\rangle=e^{-i\theta L}|0^n\rangle$. Then
\[
\langle \theta|H|\theta\rangle=\langle 0^n|H|0^n\rangle+\theta \sum_{\{j_1,\dots,j_t\}\in S}|\langle 0^n|[P_{j_1}P_{i_2}\dots P_{j_t}, H]|0^n\rangle| +\mathrm{Err},
\]
where the higher order terms $\mathrm{Err}$ can be bounded as
\begin{align}
|\mathrm{Err}| &=\left|\sum_{m\geq 2} \frac{i^m \theta^m}{m!} \langle 0^n|[L,H]_m|0^n\rangle\right|\nn\\
&\leq |E|\sum_{m\geq 2} \frac{\theta^m }{m!}\left(2kd \binom{k}{t-1}\right)^m\nn\\
&\leq \left(2kd\binom{k}{t-1}\right)^2\theta^2|E| e^{2kd\binom{k}{t-1}\theta}.\label{eq:rr31}
\end{align}
In the second line, we used the fact that 
$[L,H]_m=\sum_R [L,h_R]_m$ and each term $[L,h_R]$ can be expanded as a sum of at most $\left(kd\binom{k}{t-1}\right)^m$ non-zero terms each of norm at most $2^m$. This is because the operators $P_{j_1},P_{j_2}, \dots , P_{j_t}$ commute with each other for different choices of $\{j_1,j_2,\dots,j_t\}$ and only those that overlap with the support of $h_R$ may contribute. The number of such operators (i.e. $|S \cap \supp(h_R)|$) can be bounded by $kd\binom{k}{t-1}$ as follows: There are at most $k$ vertices in $\supp(h_R)$ and each vertex is in the support of $\leq d$ other terms in the Hamiltonian. Given a vertex $j\in \supp(h_R)$ and an overlapping Hamiltonian term $h_{R'}$ such that $j\in \supp(h_{R'})$, there are $\binom{k}{t-1}$ choices of vertices $\{j_1,j_2,\dots,j_t\}\subseteq \supp(h_{R'})$ that contain $j$. Hence, from the definition of set $S$ follows that $|S \cap \supp(h_R)|\leq kd\binom{k}{t-1}$ (one can obtain tighter bounds using $\bra{0}P_j\ket{0}=0$). Now define
\begin{align}
\beta=\frac{1}{|E|} \sum_{\{j_1,\dots,j_t\}\in S}|\langle 0^n|[P_{j_1} P_{j_2}\dots P_{j_t}, H]|0^n\rangle|.
\label{eq:rr34}
\end{align}
It holds that $\beta\leq  2 \binom{k}{t}$. To see this, note that $|\langle 0^n|[P_{j_1} P_{j_2}\dots P_{j_t}, H]|0^n\rangle|\leq \sum_R |\langle 0^n|[P_{j_1} P_{j_2}\dots P_{j_t}, h_R]|0^n\rangle|$. Using the assumption $\bra{0}P_j\ket{0}=0$, it follows that the only choices of vertices $\{j_1,\dots,j_t\}$ that may contribute in $|\langle 0^n|[P_{j_1} P_{j_2}\dots P_{j_t}, h_R]|0^n\rangle|$ are those which are fully contained in $\supp(h_R)$. The number of such vertices is bounded by $\binom{k}{t}$. Using $|\langle 0^n|[P_{j_1} P_{j_2}\dots P_{j_t}, h_R]|0^n\rangle|\leq 2$, we arrive at the claimed bound $\beta\leq  2 \binom{k}{t}$. We have
\[
\langle \theta|H|\theta\rangle=\langle 0^n|H|0^n\rangle+|E|\left(\theta\beta- \left(2kd\binom{k}{t-1}\right)^2\theta^2e^{2kd\binom{k}{t-1}\theta}\right).
\]
Choosing $\theta=O\left(\frac{\beta}{ k^2d^2\binom{k}{t-1}^2}\right)$ gives
\begin{align}
\langle \theta|H|\theta\rangle \geq 0^n|H|0^n\rangle+\Omega\left(\frac{|E|\beta^2}{k^2d^2\binom{k}{t-1}^2}\right).\label{eq:betabound}
\end{align}
Now let $\beta_1$ be given by Eq.~\eqref{eq:rr34} with $P_{j_1}P_{j_2}\dots P_{j_t}=X_{j_1}\otimes X_{j_2}\otimes \dots \otimes X_{j_t}$ for all $\{j_1,\dots,j_t\}\in S$. Define $|\hat{e}_{j_1,\dots,j_t}\rangle=X_{j_1}\otimes X_{j_2}\otimes \dots \otimes X_{j_t}|0^n\rangle$ and the operator $$p=\left(\begin{matrix}
  0 & e^{-i\frac{\pi}{2t}}\\
  e^{i\frac{\pi}{2t}} & 0
\end{matrix}\right).$$ Let $\beta_2$ be given by Eq.~\eqref{eq:rr34} with $P_{j_1}P_{j_2}\dots P_{j_t}=p_{j_1}\otimes p_{j_2}\otimes \dots \otimes p_{j_t}$ for all $\{j_1,\dots,j_t\}\in S$. Then, one can see that
\begin{align*}
\frac{\beta_1+\beta_2}{2}& =\frac{1}{|E|}\sum_{\{j_1,\dots,j_t\}\in S} \left| \mathrm{Im}(\langle \hat{e}_{j_1,\dots,j_t}|H|0^n\rangle)\right|+  \left| \mathrm{Re}(\langle \hat{e}_{j_1,\dots,j_t}|H|0^n\rangle)\right|\\ 
&\geq\frac{1}{|E|}\sum_{\{j_1,\dots,j_t\}\in S} d\cdot\left(\frac{\left| \mathrm{Im}(\langle \hat{e}_{j_1,\dots,j_t}|H|0^n\rangle)\right|^2}{d^2}+  \frac{\left| \mathrm{Re}(\langle \hat{e}_{j_1,\dots,j_t}|H|0^n\rangle)\right|^2}{d^2}\right)\\
&= \frac{1}{d|E|} |\langle 0^n |H Q_t H|0^n\rangle|\\
& \geq \frac{1}{k d |E|}\mathrm{Var}(H).
\end{align*}
This implies either $\beta_1$ or $\beta_2$ is larger than $\frac{1}{k d |E|}\mathrm{Var}(H)$. By plugging this into Eq.~\eqref{eq:betabound} we arrive at
\[
\langle \theta|H|\theta\rangle \geq \langle 0^n|H|0^n\rangle+\Omega\left(\frac{\Var(H)^2}{k^4d^4 \binom{k}{t-1}^2|E|}\right).
\]
This bound is minimized by allowing $t=\lceil{k/2+1\rceil}$ which results in the following overall lower bound on the improvement to the energy of the input state $\ket{0^n}$:
\begin{align}
\langle \theta|H|\theta\rangle \geq \langle 0^n|H|0^n\rangle+\Omega\left(\frac{\mathrm{Var}(H)^2}{ 2^{O(k)}d^4|E|}\right).\label{eq:rr40}
\end{align}
We note that the $\Omega(1/d^4)$ dependence of \eqref{eq:rr40} on the degree is quadratically worse than the bound we obtained for $2$-local Hamiltonians in \thmref{prodvar}; It would be interesting to recover the $\Omega(1/d^2)$ scaling in this case.

\section{Local classical algorithms}

\begin{theorem}\label{thm:Local classical algorithms}
Consider a two-local Hamiltonian $H=\sum_{\{i,j\}\in E} h_{ij}$ where $G=(V,E)$ is a $d$-regular triangle-free graph.  There is an efficient randomized algorithm that computes a product state $\ket{v}=\otimes_{i=1}^n \ket{v_i}$ satisfying
\begin{align}
   \expec_{v} \bra{v}H\ket{v}\geq \frac{1}{4}\mathrm{Tr}(H)+\Omega\left( \frac{\mathrm{quad}(H)}{\sqrt{d}}\right).\label{eq:r9}
\end{align}
\end{theorem}
Note that in Eq.~\eqref{eq:r9} the first term on the right hand side is equal to the expected energy of $H$ with respect to a random state $\rho=I/2^n$.

\begin{proof}[Proof of \thmref{Local classical algorithms}]
It will be convenient to work in a local Pauli basis $X,Y$ or $Z$ chosen at random and independently for each qubit.  We write $h_{ij}$ in this randomly chosen product basis. Let us define $w_{ij}=\mathrm{Tr}(h_{ij})/4$ and
\[
u_{ij}^{x}=\frac{1}{4}\mathrm{Tr}( h_{ij} X_i X_j)\quad u_{ij}^{y}=\frac{1}{4}\mathrm{Tr}( h_{ij} Y_i Y_j)\quad u_{ij}^{z}=\frac{1}{4}\mathrm{Tr}( h_{ij} Z_i Z_j).
\]
Due to the random choice of basis we have
\begin{equation}
\mathbb{E}[(u_{ij}^{a})^2]=\frac{1}{9}\mathrm{quad}(h_{ij}) \qquad a\in \{x,y,z\}.
\label{eq:eu}
\end{equation}

Following \cite{BarakMORRSTVW15,hastings2019BoundedDepth}, we start with a random i.i.d assignment of pure product $\ket{v}=\otimes_{i=1}^n \ket{v_i}$ states to the vertices. We then select a subset of vertices $A$ uniformly at random. For any vertex $i\in  A$, let $N(i)=\{\{i,j\}\in E: j \notin A\}$ be the neighboring edges that contain exactly one vertex in $A$ (i.e. vertex $i$). The remaining edges that are not in $\cup_{i\in A} N(i)$ either connect two vertices that are not in $A$ or connect two vertices in $A$. We denote the former by $M$ and the latter by $M'$.

The initial random pure state at each vertex $\r_i$ can be represented by $\r_i=\frac{1}{2}(\iden+ r^x_i X_i+r^y_i Y_i+r^z_i Z_i)$, where $(r^x_i,r^y_i,r^z_i) \in \bbR^3$ is the Bloch vector with norm $|r^x_i|^2+|r^y_i|^2+|r^z_i|^2=1$. For a vertex $i\in A$, the total energy of the edges $N(i)$ is given by
\begin{align}
    \sum_{j:\{i,j\}\in N(i)} \Tr[h_{ij}\r_i \ot \r_j]=\sum_{j:\{i,j\}\in N(i)} w_{ij}+ \sum_{j:\{i,j\}\in N(i)}(u_{ij}^x r^x_i r^x_j +u_{ij}^y r^y_i r^y_j+u_{ij}^z r^z_i r^z_j)+\sum_{j:\{i,j\}\in N(i)}D_{ij}(\vec{r}_i,\vec{r}_j)\label{eq:first2}\end{align}
where
\[
D_{ij}(\vec{r}_i,\vec{r}_j)=\sum_{a\neq b} c^{ab}_{ij} r_i^{a}r_j^b+\sum_{a\in \{x,y,z\}}(d^a_{ij}r_i^{a}+e^a_{ij} r_j^a).
\]
for some coefficients $c^{ab}_{ij},d^{a}_{ij}, e^{a}_{ij}$. Using Cauchy–Schwarz inequality, we see that the first two terms in Eq.~\eqref{eq:first2} can be maximized by applying a local unitary on each vertex $i\in A$ which rotates the state $\r_i$ to a state $\tilde{\r}_i$ with the Bloch vector 
\begin{align}
R^a_i=(\sum_{j:\{i,j\}\in N(i)}u^a_{ij} r_j^a)\bigg((\sum_{j:\{i,j\}\in N(i)}u_{ij}^x  r^x_j)^2 +(\sum_{j:\{i,j\}\in N(i)}u_{ij}^y r^y_j)^2+(\sum_{j:\{i,j\}\in N(i)}u_{ij}^z r^z_j)^2\bigg)^{-1/2} \label{eq:rr37}
\end{align}
for $a\in \{x,y,z\}$. When the denominator of Eq.~\eqref{eq:rr37} is zero, the vector $\vec{R}$ is chosen uniformly at random. Hence, we get
\begin{align}
    \sum_{j:\{i,j\}\in N(i)} \Tr[h_{ij}\tilde{\r}_i \ot \r_j]&= \sum_{\{i,j\}\in N(i)} w_{ij}+ \bigg((\sum_{j:\{i,j\}\in N(i)}u_{ij}^x  r^x_j)^2 +(\sum_{j:\{i,j\}\in N(i)}u_{ij}^y r^y_j)^2+(\sum_{j:\{i,j\}\in N(i)}u_{ij}^z r^z_j)^2\bigg)^{1/2}\\\nn
&+\sum_{j:\{i,j\}\in N(i)}D_{ij}(\vec{R}_i,\vec{r}_j) ,\nn
\end{align}
A property of this construction is that $\expec[R_i^a]=0$ for $a\in \{x,y,z\}$ and $R_i$, $R_j$ are independent of each other for $\{i,j\}\in M'$. This follows from the triangle-freeness, the definition of the set $N(i)$, and the initial uniform i.i.d. distribution of the state of vertices. Moreover, we have
\[
\mathbb{E}[r_j^{a}r_{k}^{b}]=\mathbb{E}[R_j^{a}r_{k}^{b}]=0 \quad \text{whenever}\quad a\neq b\quad.
\]

Using these observations, the expected value of the total energy after the local improvements is 
\begin{align}
    & \expec\left[\sum_{\{i,j\}\in M'} \Tr[h_{ij}\tilde{\r_i} \ot \tilde{\r_j}]\right] +\expec\left[\sum_{\{i,j\}\in M} \Tr[h_{ij}\r_i \ot \r_j]\right]+\expec\left[\sum_{i\in A} \sum_{j:\{i,j\}\in N(i)}\Tr[h_{ij}\tilde{\r_i} \ot \r_j]\right] \nn\\
     &=\sum_{\{i,j\}\in E} w_{ij}+ \expec \left[\sum_{i\in A}\bigg((\sum_{j:\{i,j\}\in N(i)}u_{ij}^x  r^x_j)^2 +(\sum_{j:\{i,j\}\in N(i)}u_{ij}^y r^y_j)^2+(\sum_{j:\{i,j\}\in N(i)}u_{ij}^z r^z_j)^2\bigg)^{1/2} \right]\nn \\
     &\geq \sum_{\{i,j\}\in E} w_{ij} +\expec\left[\sum_{i\in A} |\sum_{j:\{i,j\}\in N(i)}u_{ij}^z r^z_j|\right]\label{eq:r6}
\end{align}
The first term $\sum_{\{i,j\}\in E} w_{ij}$ corresponds to the expected energy when the product states are chosen uniformly at random. The second term is a lower bound on the improvement achieved by the local updates which we now show is at least $\Omega(|E|/\sqrt{d})$.

For a fixed choice of the set $A\subseteq V$, we define the random variable $\xi_i=  \sum_{j:\{i,j\}\in N(i)}u_{ij}^z r^z_j$. Using the ``second moment method,'' for $t\in [0,1]$, we get 
\begin{align}
      \Pr\left[|\xi_i|\geq t\sqrt{\expec[\xi^2]}\right]\geq (1-t^2)^2 \frac{\expec[\xi^2]^2}{\expec[\xi^4]} \label{eq:r5}
\end{align}
 One way to sample uniform pure states over Bloch sphere is to uniformly draw $\phi\sim [0,2\pi]$, $r_j^z \sim [-1,1]$ and set $r_j^x=\sqrt{1-(r_j^z)^2} \cos \phi$ and $r_j^y=\sqrt{1-(r_j^z)^2} \sin \phi$. Given this we have $\expec[r_j^z]=0$,  $\expec[(r_j^z)^2]=1/3$, $\expec[(r_j^z)^3]=0$, and $\expec[(r_j^z)^4]=1/5$. Hence, using Corollary~9.6 of \cite{o2014booleananalysis}, we have $\expec[\xi^4]\leq 9\cdot \expec[\xi^2]^2$. Plugging this in \eqref{eq:r5} implies that for a fixed choice of the set $A$, the expectation with respect to the random distribution of the initial product states for an arbitrary choice of $t\in[0,1]$ is
\begin{align}
     \expec\left[|\sum_{j:\{i,j\}\in N(i)}u_{ij}^z r^z_j|\right]&\geq \frac{1}{9}\cdot t(1-t^2)^2\cdot \expec[\xi^2]^{1/2}\nn\\  
     &= \frac{1}{9}\cdot t(1-t^2)^2\cdot \left( \sum_{j:\{i,j\}\in N(i)} (u_{ij}^z)^2 \expec\left[(r^z_j)^2\right]\right)^{1/2}\nn\\ 
    &\geq \frac{1}{9\sqrt{3}}\cdot t(1-t^2)^2\cdot \left(\sum_{j:\{i,j\}\in N(i)} (u_{ij}^z)^2\right)^{1/2}.\label{eq:r7}
\end{align}
Finally, we calculate the expectation with respect to the set $A\subseteq V$. Note that the set $N(i)$ is also a random variable determined by the set $A$. Conditioned on the event that the vertex $i\in A$ and using Theorem~9.24 of \cite{o2014booleananalysis}, we have
$$\Pr\left[ \sum_{j:\{i,j\}\in N(i)} (u_{ij}^z)^2\geq \frac{1}{2} \sum_{j:\{i,j\}\in E} (u_{ij}^z)^2\right]\geq \frac{1}{4e^2}.$$
Thus, we get 
\begin{align}
    \expec\left[\sum_{i\in A}\left(\sum_{j:\{i,j\}\in N(i)} (u_{ij}^z)^2\right)^{1/2}\right]&=\frac{1}{8\sqrt{2}e^2}\sum_{i\in V}\left(\sum_{j:\{i,j\}\in E} (u_{ij}^z)^2\right)^{1/2}\nn\\
    &=\frac{1}{4\sqrt{2}e^2}\cdot \frac{1}{\sqrt{d}}\sum_{\{i,j\}\in E} (u_{ij}^z)^2.
    \end{align}
Finally, taking the expectation over the random choice of local basis and using Eq.~\eqref{eq:eu} we get
 \begin{align}
\expec\left[\sum_{i\in A}\left(\sum_{j:\{i,j\}\in N(i)} (u_{ij}^z)^2\right)^{1/2}\right]
    &\geq \frac{1}{36\sqrt{2}e^2}\cdot \frac{1}{\sqrt{d}}\sum_{\{i,j\}\in E}\mathrm{quad}(h_{ij})\nn\\
    &\geq \frac{1}{36\sqrt{2}e^2}\cdot \frac{1}{\sqrt{d}}\mathrm{quad}(H)
\end{align}
We arrive at \eqref{eq:r9} by plugging this into \eqref{eq:r7} and using \eqref{eq:r6}.

\end{proof}

\end{document}